%% file: main.tex
\documentclass[twoside]{article}

\usepackage[accepted]{aistats2025}
\input{macros}

\usepackage[round]{natbib}

\begin{document}

\twocolumn[

\aistatstitle{Bayesian Inference Under Differential Privacy With Bounded Data}

\aistatsauthor{Zeki Kazan \And Jerome P.~Reiter}

\aistatsaddress{Duke University \And  Duke University}]

\begin{abstract}
  We describe Bayesian inference for the parameters of Gaussian models of bounded data protected by differential privacy. 
  Using this setting, we demonstrate that analysts can and should take constraints imposed by the bounds into account when specifying prior distributions. Additionally, we provide theoretical and empirical results regarding what classes of default priors produce valid inference for a differentially private release in settings where substantial prior information is not available. We discuss how these results can be applied to Bayesian inference for regression with differentially private data.
\end{abstract}

\section{INTRODUCTION}

Differential privacy (DP) \citep{dwork2006calibrating} is the current gold standard for protecting individual privacy when performing data releases involving confidential data. A common method of ensuring a release satisfies DP is to add calibrated random noise to sufficient statistics of the data. Once the noisy statistics are released, analysts in theory can use them to perform statistical tasks such as the estimation of unknown parameters of interest, prediction of new data values, and quantification of the uncertainty around these quantities. Performing these tasks in a principled way, however, can be nontrivial and is an area of active research.

Bayesian methods are natural for statistical inference after a DP release. By sampling from posterior distributions, analysts can perform estimation and prediction tasks automatically with ``built-in'' uncertainty quantification. As examples, \cite{bernstein2018differentially} and \cite{ju2022data} use Gibbs sampling to estimate differentially private posterior distributions, the former for exponential families and the latter for likelihoods satisfying a record additivity condition. \cite{gong2022exact} uses approximate Bayesian computation for settings where a perturbation mechanism is applied. Other works investigate Bayesian estimation of subset proportions \citep{li2022bayesian},  regression coefficients \citep{bernstein2019differentially}, and Poisson factorization \citep{schein2019locally}. These methods employ variations of the following scheme: repeatedly impute plausible values of the nonprivate statistics based on their private counterparts, and obtain samples of the parameters of interest via a nonprivate Bayesian analysis.

DP is related to Bayesian methods more generally. 
Releasing a sample from a posterior distribution can be used to achieve DP \citep{dimitrakakis2017differential, geumlek2017renyi, jewson2024differentially, minami2016differential, wang2015privacy, zhang2023dp, zhang2016differential}. Another line of work examines Bayesian semantics for DP and their relationship to disclosure risk assessment \citep{dwork2006calibrating, kasiviswanathan2014semantics, kazan2023prior, kazan2022assessing, kifer2022bayesian,kifer2014pufferfish, mcclure2012differential}. 

Existing methods for Bayesian inference under DP generally presume that the analyst sets an informative prior distribution. In practice, however, the analyst may not have sufficient prior knowledge to set an accurate, informative prior. Additionally, methods in the literature typically rely on prior distributions with unbounded support. Often, however, the confidential data are assumed to be constrained in some way; for example, values lie within an interval $[a,b]$ \citep{bernstein2019differentially, du2020differentially, sheffet2017differentially, wang2018revisiting, zhang2012functional}. 
Arguably, samples from the posterior distribution should reflect this constraint as well. Yet, existing methods for Bayesian inference under DP  can produce estimates or predictions outside the feasible range.  This may lead analysts to resort to ad hoc fixes such as clipping estimates or predictions to be at the bounds of the interval, which can lead to undesirable and unpredictable  behavior.

We address incorporating 
constraints in Bayesian inference under DP, including implications for prior specification. 
Our primary contributions include:
\begin{enumerate}
\item We propose an improved Gibbs sampler for Gaussian models to enable posterior inference after the release of DP statistics. The sampler uses exact forms of all full conditional distributions, avoiding 
approximations
that are typically used in other Bayesian post-processing inferential approaches. 
    \item We propose a method 
    to incorporate constraints into a Bayesian inference procedure after a DP data release that is computationally efficient, simple, and generalizes to a range of settings. We illustrate its benefits for estimation and prediction.
    \item We present a first-of-its-kind examination of default prior distributions analysts might consider in lieu of highly informative proper priors, and present theoretical results demonstrating the need for caution when considering weakly informative priors in a DP analysis.
\end{enumerate}

Throughout, we use the univariate Gaussian setting as a motivating case study. This setting is among the most important and commonly encountered in statistics but is relatively understudied in the DP literature, with only a few works providing frequentist methods for interval estimation \citep{d2015differential, du2020differentially, ferrando2022parametric, karwa2017finite} or hypothesis testing \citep{couch2019differentially, kazan2023test, pena2022differentially}.

\section{BACKGROUND AND SETTING}

We first briefly review 
DP and several of its properties. We then describe the DP univariate Gaussian setting and present our novel Gibbs sampling scheme.

\subsection{Differential Privacy}

As described by \cite{dwork2006calibrating}, a release mechanism $\M$ satisfies $\varepsilon$-DP if for all data sets $x,x'$ that differ in only one row and all $S \subseteq Range(\M)$,
\begin{equation} \label{eq:dp}
    P[\M(x) \in S] \leq e^{\varepsilon}\, P[\M(x') \in S].
\end{equation}
DP has many desirable properties \citep{dwork2006calibrating}. Here, we make use of \textit{post-processing.} For any data set $x$, if $\M(x)$ satisfies $\varepsilon$-DP and $g$ is not a function of $x$, then $g(\M(x))$ satisfies $\varepsilon$-DP. We also make use of \textit{composition.} For any data set $x$, if $m_1 = \M_1(x)$ satisfies $\varepsilon_1$-DP and $m_2 = \M_2(x, m_1)$ satisfies $\varepsilon_2$-DP, then $\M(x) = (m_1, m_2)$ satisfies $(\varepsilon_1 + \varepsilon_2)$-DP.

For any function $f$ applied to any data set $x$, the sensitivity is defined as $\Delta f = \max_{x, x'} |f(x) - f(x')|$ for all data sets $x,x'$ that differ in only one row.  The sensitivity of $f$ is used in the \textit{Laplace Mechanism,} which is a common method for achieving DP. If a function $f(x)$ has sensitivity $\Delta f$, then the mechanism $\M(x) = \textsf{Lap}(f(x), \Delta f/\varepsilon)$ satisfies $\varepsilon$-DP.

\subsection{The Univariate Gaussian Setting} \label{sec:gibbs}

Let $Y_1, \ldots, Y_n$ be confidential, scalar data values for $n$ individuals. The $Y_i$ are contained in some interval $[a,b]$
set by the data curator that accords with scientific understanding of the phenomenon under study and, to maintain the DP guarantee, is independent of the realized  data values. 
 As examples, test scores may be bounded between $[a,b] = [0, 100]$ and IQ scores between $[a, b] = [40, 250]$.  This interval is made public by the data curator.  
 
The data curator releases private versions of the sample mean and variance, i.e., 
$\bar{Y} = \sum_{i=1}^n Y_i / n$ and $S^2 = \sum_{i=1}^n (Y_i - \bar{Y})^2/(n-1)$. As shown in \cite{du2020differentially}, $\bar{Y}$ and $S^2$ have sensitivities $(b-a)/n$ and $(b-a)^2/n$, respectively. To achieve $\varepsilon$-DP, we presume the data curator releases $\bar{Y}^* \sim \textsf{Lap}(\bar{Y}, (b-a)/(\varepsilon_1 n))$ and $S^{2*} \sim \textsf{Lap}(S^2, (b-a)^2/(\varepsilon_2 n))$, where $\varepsilon_1 + \varepsilon_2 = \varepsilon$. 

We make the simplifying assumptions $a=0$ and $b=1$. We can do so without loss of generality, as described in Theorem \ref{thm:scale}, which is proved in Appendix \ref{app:scale}.
\begin{thm} \label{thm:scale}
    Let $Y_1, \ldots, Y_n \in [a,b]$ and let $\tilde{Y}_i = (Y_i - a)/(b - a) \in [0,1]$. Let $\bar{Y}$ and $S^2$ be  
    the sample mean and variance for $\{Y_i\}$ and let $\tilde{\bar{Y}}$ and $\tilde{S}^2$ be the sample mean and variance for $\{ \tilde{Y}_i\}$. Suppose each statistic is released via the Laplace Mechanism under $\varepsilon$-DP and denote the DP statistics $\bar{Y}^*$, $S^{2*}$, $\tilde{\bar{Y}}^*$, and $\tilde{S}^{2*}$. Then, $\bar{Y}^* \overset{d}{=} (b-a)\tilde{\bar{Y}}^* + a$ and $S^{2*} \overset{d}{=} (b-a)^2 \tilde{S}^{2*}$.
\end{thm}

A data analyst who only has access to $\bar{Y}^*$ and $S^{2*}$, as well as the particulars of the DP mechanism used to create them, considers these released summary statistics as noisy versions of $\bar{Y}$ and $S^2$, which themselves are estimates of underlying population parameters $\mu$ and $\sigma^2$. To make inferences for these population parameters, we presume the analyst posits a model for the underlying confidential data, namely  
$Y_i \sim \norm(\mu, \sigma^2)$ independently for all $i$. While theoretically this model has support outside the interval $[a,b]$, in practice it accurately describes data distributions when the observed values are symmetric with an empirical range not near the bounds. Indeed, data analysts  frequently use normal models for bounded data; for example, normal models are typical for quantities like proportions and standardized test scores, which have known ranges. They also are used, for example, for quantities like health measurements (e.g., blood pressures, birth weights) and stock market returns, which can be bounded based on extensive historical observations and scientific considerations.

We consider an analyst who follows the Bayesian paradigm to obtain posterior inferences about $(\mu, \sigma^2)$ given $(\bar{Y}^*, S^{2*})$.  We consider first the case of an analyst who incorporates prior beliefs into their analysis via the convenient, informative prior distribution,
\begin{align} \label{eq:prior}
    \sigma^2 \sim \textsf{IG}\left(\nu_0/2, \nu_0\sigma_0^2/2 \right), \hspace{1.25mm}  \mu \mid \sigma^2 \sim \norm\left(\mu_0, \sigma^2/\kappa_0\right)
\end{align}
for analyst-specified hyperparameters $(\nu_0, \kappa_0, \mu_0, \sigma_0^2).$  Here, $\textsf{IG}$ denotes the inverse-gamma distribution. The prior distribution in \eqref{eq:prior} is  conjugate for the normal model, which facilitates computation.  Using \eqref{eq:prior}, the posterior distribution $p(\mu, \sigma^2 | \bar{Y}^*, S^{2*})$ can be estimated using Markov chain Monte Carlo techniques. 

\cite{bernstein2018differentially, bernstein2019differentially} present Gibbs samplers for similar inferential tasks, in which one samples from the full conditional of each unknown parameter in $(\mu, \sigma^2, \bar{Y}, S^2)$ at each iteration.  However, these samplers rely on an
approximation to the full conditional distribution of $S^2$, namely that it is normally distributed.  Clearly this is not the case,
since $S^2$ must be positive.
Thus, we come to our first innovation: we derive a Gibbs sampler that 
utilizes the exact full conditional distribution of  $S^2$, thereby eschewing the normal approximation for $S^2$. 
Specifically, when $\varepsilon_2 < 2$, the full conditional for $S^2$ is a distribution we call a truncated gamma mixture (TGM).\footnote{For the distribution of $S^2$ to be a TGM, we require $\sigma^2 < (n-1)/(2n\varepsilon_2)$. We show in Corollary \ref{cor:sigma_sq} that in this setting we can guarantee $\sigma^2 \leq 1/4$, which implies $\sigma^2 < (n-1)/(2n\varepsilon_2)$ automatically when $\varepsilon_2 < 2$. Within the Gibbs sampler, we reject and resample any draws for which $\sigma^2 \geq (n-1)/(2n\varepsilon_2)$.}
\begin{defi}
    $X \sim \textsf{TGM}(\alpha, \beta, \lambda, \tau)$ for $\alpha > 0$, $\tau \in \R$, and  $\beta > \lambda \geq 0$ if its distribution is as follows. When $\tau \leq 0$, $X \sim \textsf{Gamma}(\alpha, \beta + \lambda)$ and when
    when $\tau > 0$, it has probability density function
    \begin{align}
        p(x) = \begin{cases}
            \pi_1 \frac{(\beta - \lambda)^\alpha}{\gamma(\alpha, (\beta - \lambda)\tau)} x^{\alpha-1} e^{-(\beta - \lambda)x}, & \mbox{if } x \leq \tau; \\
            \pi_2 \frac{(\beta + \lambda)^\alpha}{\Gamma(\alpha, (\beta + \lambda)\tau)} x^{\alpha-1} e^{-(\beta + \lambda)x} & \mbox{if } x > \tau,
        \end{cases}
    \end{align}
    where
    \begin{align}
        \begin{matrix}
            \pi_1 &= \frac{e^{-\lambda \tau} \frac{\gamma(\alpha, (\beta - \lambda)\tau)}{(\beta - \lambda)^\alpha}}{e^{-\lambda \tau} \frac{\gamma(\alpha, (\beta - \lambda)\tau)}{(\beta - \lambda)^\alpha} + e^{\lambda \tau} \frac{\Gamma(\alpha, (\beta + \lambda)\tau)}{(\beta + \lambda)^\alpha}}, \\
            \pi_2 &= \frac{e^{\lambda \tau}\frac{\Gamma(\alpha, (\beta + \lambda)\tau)}{(\beta + \lambda)^\alpha}}{e^{-\lambda \tau} \frac{\gamma(\alpha, (\beta - \lambda)\tau)}{(\beta - \lambda)^\alpha} + e^{\lambda \tau} \frac{\Gamma(\alpha, (\beta + \lambda)\tau)}{(\beta + \lambda)^\alpha}}. 
        \end{matrix}
        \label{eq:weights}
    \end{align}
\end{defi}
Here, $\gamma(\alpha, x)$ and $\Gamma(\alpha, x)$ are the lower and upper incomplete gamma functions. We provide an algorithm for sampling from the TGM in Appendix \ref{app:TGM}.

Derivations for all full conditionals for the Gibbs sampler are provided in Appendix \ref{app:Gibbs}. The full conditionals for $\mu$ and $\sigma^2$ are identical to those of the nonprivate setting. To obtain the full conditional for $\bar{Y}$, we follow the strategy of \cite{bernstein2018differentially} based on the Bayesian LASSO \citep{park2008bayesian} and use the fact that $\bar{Y}^* \sim \textsf{Lap}(\bar{Y}, 1/(\varepsilon_1 n))$ is equal in distribution to $\bar{Y}^* \sim \norm(\bar{Y}, \omega^2)$ for $\omega^2 \sim \textsf{Exp}(\varepsilon_1^2 n^2 / 2)$. The full conditional for $\bar{Y}^*$ is Gaussian, the full conditional for $1/\omega^2$ is inverse-Gaussian, and the full conditional for $S^2$ is a TGM. The parameters of the full conditionals and an explicit algorithm for our Gibbs sampler are provided in Appendix \ref{app:full_cond}.

A major advantage of this sampler is its runtime. The computational complexity of this sampler is $\mathcal{O}(T)$, where $T$ is the number of Gibbs iterations. The runtime does not scale with $n$. The Gibbs sampler of  \cite{bernstein2018differentially}  
applied to the univariate Gaussian setting also has computational complexity $\mathcal{O}(T)$. This sampler, however, requires a Gaussian approximation for $S^2$, which, particularly when $\sigma^2$ is near zero or $n$ is small, can lead to unreliable results. The sampler of \cite{ju2022data}  does not require any approximations; however, it has computational complexity $\mathcal{O}(nT)$, making it potentially infeasible for large $n$. Our proposed sampler achieves both the exact sampling of \cite{ju2022data} and the $\mathcal{O}(T)$ complexity of \cite{bernstein2018differentially}.

\section{ENFORCING CONSTRAINTS} \label{sec:bounds}

The modeling in Section \ref{sec:gibbs} does not account for the fact that the data values lie within some $[a, b]$.  In this section, we show that analysts who incorporate the constraints on the data values can reduce the uncertainty in estimates of parameters of interest ``for free.'' They need not make additional assumptions; they merely determine how conditions on the data values constrain the model used for analysis.  

There are several ways an analyst could incorporate these constraints. In this work, we introduce implied constraints on all unknown parameters as part of the modeling and sampling.  This leads to straightforward estimation methods that can be readily generalized beyond the simple normal model.  Another approach is to  put constraints directly on the likelihood function for the model. We illustrate this approach for the normal model in Appendix \ref{app:constraints_comp}. Ultimately, we find that our approach is computationally advantageous 
and can offer more reliable repeated sampling performance; 
see Appendix \ref{app:constraints_comp} for details.

\subsection{Constraints for the Gaussian Setting} \label{sec:Gauss_bounds}

In the univariate Gaussian setting from Section \ref{sec:gibbs}, the constraint that $Y_i \in [0,1]$ affects the model in two ways. Firstly, the bound constrains the parameters $\mu$ and $\sigma^2$, since bounded random variables have bounded moments. Secondly, the bound constrains the sufficient statistics $\bar{Y}$ and $S^2$, since samples of bounded random variables have bounded means and variances. We demonstrate how we establish these bounds in a manner that can be easily incorporated into the Gibbs sampler from Section \ref{sec:gibbs}. See Appendix \ref{app:bound} for proofs.

First, Theorem \ref{thm:par_bounds} establishes conditional bounds for $\mu \mid \sigma^2$ and $\sigma^2 \mid \mu$. These bounds follow from monotonicity of expectation. As a result, since the data values are assumed to be bounded, we may replace the unbounded normal-inverse gamma prior from (\ref{eq:prior}) with a prior of the same form truncated to be within the feasible region. This yields a posterior of the same form as in Section \ref{sec:gibbs} truncated to the feasible region.

\begin{thm} \label{thm:par_bounds}
    Let $Y_i \in [0,1]$ have moments $E[Y_i] = \mu$ and $V[Y_i] = \sigma^2$.
    It follows that $\sigma^2 \in [0, \mu(1-\mu)]$ and $\mu \in [1/2 - \sqrt{1/4 - \sigma^2}, 1/2 + \sqrt{1/4 - \sigma^2}]$.
\end{thm}

\begin{cor} \label{cor:sigma_sq}
    If $Y_i \in [0,1]$, then $V[Y_i] = \sigma^2 \leq 1/4$.
\end{cor}

Theorem \ref{thm:stat_bounds} establishes conditional bounds for $\bar{Y} \mid S^2$ and $S^2 \mid \bar{Y}$. This yields a posterior of the same form as in Section \ref{sec:gibbs} truncated to the feasible region.

\begin{thm} \label{thm:stat_bounds}
    Let $Y_1, \ldots, Y_n$ be such that each $Y_i \in [0,1]$. Then $S^2 \in [0, n/(n-1)\bar{Y}(1-\bar{Y})]$ and $\bar{Y} \in [1/2 - \sqrt{1/4 - (n-1)/n \cdot S^2}, 1/2 + \sqrt{1/4 - (n-1)/n \cdot S^2}]$.
\end{thm}

The full conditionals under these constraints are described in Appendix \ref{app:full_cond}.

\subsection{Example: The Blood Lead Dataset} \label{sec:bounds_lead}

We now demonstrate the effect of incorporating these constraints using genuine data.\footnote{Code to reproduce all experiments is in Appendix \ref{app:code} and at \url{https://github.com/zekicankazan/dp_priors}.}
All computations are performed on the $[0,1]$ scale, but all quantities are converted back to the original scale for clarity of presentation. Interval estimates are those with highest posterior density (HPD).\footnote{HPD intervals are preferred over quantile-based intervals when the posterior distribution is skewed \citep{kruschke2014doing}. HPD intervals can be computed with standard tools, such as the \texttt{HDInterval} package in \texttt{R} \citep[License:~GPL-3]{hdi}.} The following example, adapted from an introductory statistics textbook \citep[Section~7.1]{diez2012openintro}, is 
representative of a common 
statistical analysis for numerical, scalar data.

\begin{exa} \label{ex:lead}
    Researchers sampled the blood lead levels of $n = 43$ policemen assigned to outdoor work in Egypt to examine the effect of leaded gasoline on exposed individuals \citep{mortada2001study}. The sample mean was $\bar{Y} = 32.08 \, \mu g/dL$, and the sample variance was $S^2 = 16.98^2 \, \mu g^2/dL^2$. An expert advises that blood lead levels in this region are reasonably bounded above by $100 \, \mu g/dL$. Using these bounds, the researchers use the Laplace Mechanism with $\varepsilon_1 = \varepsilon_2 = 0.25$  to release noisy  statistics $\bar{Y}^* = 34.30 \, \mu g/dL$ and $S^{2*} = 47.16^2 \, \mu g^2/dL^2$. A secondary data analyst seeks inferences for $\mu$ and $\sigma^2$, the average and variance of blood lead levels of all policemen assigned to outdoor work in Egypt. She believes blood lead levels are reasonably approximated by a Gaussian distribution. To reflect prior knowledge, she  uses \eqref{eq:prior} with $(\mu_0 = 12.5, \sigma_0^2 = 3.8^2, \kappa_0 = 1, \nu_0 = 1)$, roughly representing the equivalent information of a single ``prior data point.'' 
\end{exa}

\begin{figure}[t]
    \centering
    \includegraphics {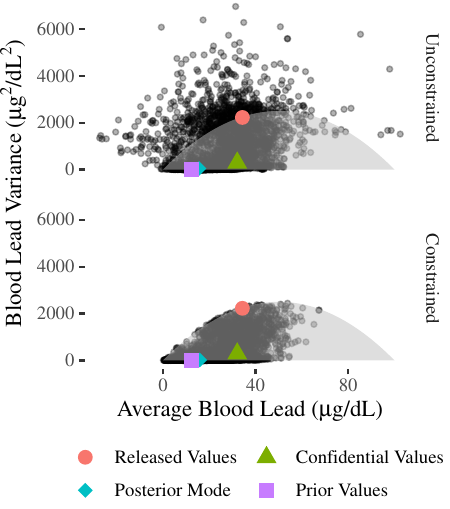}
    \caption{Joint posterior distribution for $(\mu,\sigma^2)$ in Example \ref{ex:lead}. Here, $(\bar{Y}^* = 34.30, S^{2*} = 47.17^2)$ is represented by the red circle, the unreleased $(\bar{Y} = 32.08, S^2 = 16.96^2)$ is represented by the green triangle, the analyst's posterior mode is represented by the blue diamond, and $(\mu_0 = 12.5, S^2 = 3.8^2)$ is represented by the purple square. Upper and lower panels display the posterior when constraints are and are not accounted for, respectively. The shaded area represents the feasible region for $(\mu, \sigma^2)$ from Theorem \ref{thm:par_bounds}. Plots based on 5,000 Gibbs sampler iterations.}
    \label{fig:bounding}
\end{figure}

Figure \ref{fig:bounding} presents the joint posterior distribution for $(\mu, \sigma^2)$ with and without constraints accounted for. For this particular release, $S^{2*}$ is larger than the nonprivate $S^2$ and is near the upper bound on $\sigma^2$ determined by Theorem \ref{thm:par_bounds}. Thus, more than 10\% of samples from the unconstrained analysis have infeasibly large $\sigma^2$. Additionally, the unconstrained analysis produces samples with $\mu$ either negative or above the upper bound. The constrained analysis respects the bounds and produces only reasonable samples for $\mu$ and $\sigma^2$.

In both analyses, the analyst's posterior mode is around $\hat{\mu} \approx 16 \, \mu \textup{g}/\textup{dL}$ and $\hat{\sigma}^2 \approx 5^2 \, \mu \textup{g}^2/\textup{dL}^2$. The analyses, however, have substantially different amounts of uncertainty. The 95\% HPD interval for $\sigma^2$ is $[1.2^2, 50.8^2]$ without constraints and $[1.0^2, 39.3^2]$ with constraints.  Accounting for constraints allows the analyst to rule out the region with infeasibly high $\sigma^2$, yielding a  tighter interval estimate. The 95\% HPD interval for $\mu$ is $[3.4, 48.2]$ without constraints and $[1.9, 42.0]$ with constraints. Similarly to the above, the tighter interval in the constrained analysis is due primarily to ruling out $\mu$ in the region where $\sigma^2$ is infeasibly large. In both analyses, the posterior concentration is highest near the prior values, $\mu_0$ and $\sigma_0^2$; there is much less density near the released noisy values. This indicates that the posterior inference is likely quite sensitive to the analyst's prior distribution.

The constrained posterior is not merely a truncated version of the unconstrained posterior; truncation in this way would yield inaccurate point and interval estimates. Instead, by truncating within the Gibbs sampler, we have mathematically coherent estimates. 

\section{DEFAULT PRIOR CHOICES} \label{sec:default}

As evident in Example \ref{ex:lead}, the results of a Bayesian analysis under DP can be sensitive to the
prior distribution. This is unsurprising, since the addition of noise yields data with less information about model parameters than in the public setting, making it difficult to overwhelm the information in the prior distribution. One potential fix---used in examples in the Bayesian DP inference literature---is to increase the prior variance, yielding a weakly informative prior. One also could examine the limit in which the prior variance is infinite. We show, however, that this strategy may lead to improper limiting posterior distributions, even in settings where the nonprivate posterior is proper. In such cases, weakly informative priors may yield unreliable inference and should be avoided \citep{gelman2006prior}.

\subsection{Default Priors for the Gaussian Setting} \label{sec:flat_priors}

For the unconstrained analysis in the univariate Gaussian setting, a weakly informative prior of the form in (\ref{eq:prior}) is created by making $\kappa_0$ and $\nu_0$ small. In the limit as these hyperparameters go to zero, this produces the default prior $p(\mu, \sigma^2) \propto (\sigma^2)^{-1}$. In the nonprivate setting, this is called the independent Jeffrey's prior \citep{sun2007objective}. While it does not correspond to a proper probability distribution, the posterior distribution it produces is not only proper, but also frequentist matching.
In the private setting, unfortunately, the posterior distribution produced by $p(\mu, \sigma^2) \propto (\sigma^2)^{-1}$ is not a proper probability distribution. The following result, proved in Appendix \ref{app:proper}, demonstrates this.

\begin{thm} \label{thm:IJ_proper}
    For confidential data $Y_1, \ldots, Y_n \overset{iid}{\sim} \norm(\mu, \sigma^2)$ where $\bar{Y}^* \sim \textsf{Lap}(\bar{Y}, 1/(\varepsilon_1 n))$ and $S^{2*} \sim \textsf{Lap}(S^2, 1/(\varepsilon_2 n))$ are released, if an analyst has prior $p(\mu, \sigma^2) \propto (\sigma^2)^{-1}$, then their posterior $p(\mu, \sigma^2 \mid \bar{Y}^*, S^{2*})$ is not a proper probability distribution.
\end{thm}

An alternative strategy for finding a default prior for this setting is to focus on the likelihood's Laplace portion. We show in Appendix \ref{app:lap_unif} that for a Laplace likelihood in the nonprivate setting, a uniform prior produces a proper, frequentist matching posterior distribution. A uniform prior on $(\mu, \sigma^2)$ thus may be a reasonable choice in the private setting. Theorem \ref{thm:Unif_proper}, proved in Appendix \ref{app:proper}, demonstrates that $p(\mu, \sigma^2) \propto 1$ does indeed produce a proper posterior distribution.

\begin{figure}[t]
    \centering
    \includegraphics {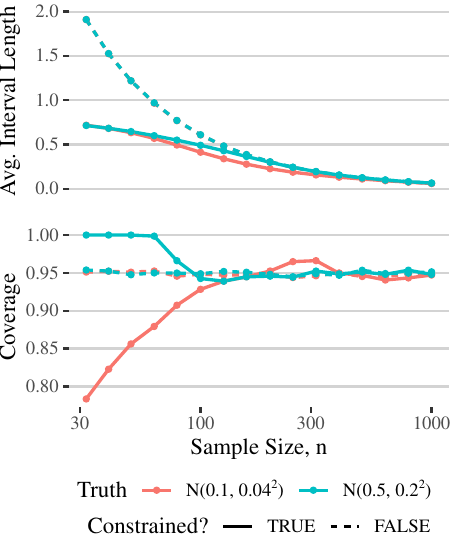}
    \caption{The average length (top) and coverage rate (bottom) of 95\% HPD intervals for $\mu$ for different $n$. Results based on 10,000 simulated datasets $Y_i \in [0,1]$ released with $\varepsilon_1 = \varepsilon_2 = 0.1$ and analyzed with prior $p(\mu, \sigma^2) \propto 1$. Analyses with constraints accounted for are solid lines and not accounted for are dashed lines. Data generating model is either $\norm(\mu = 0.1, \sigma^2 = 0.04^2)$ (red) or $\norm(\mu = 0.5, \sigma^2 = 0.2^2)$ (blue). Note that the red dashed line is below the blue dashed line. Each Gibbs sampler is run for 20,000 iterations.}
    \label{fig:coverage}
\end{figure}

\begin{thm} \label{thm:Unif_proper}
    For confidential data $Y_1, \ldots, Y_n \overset{iid}{\sim} \norm(\mu, \sigma^2)$ where $n > 3$ and $\bar{Y}^* \sim \textsf{Lap}(\bar{Y}, 1/(\varepsilon_1 n))$ and $S^{2*} \sim \textsf{Lap}(S^2, 1/(\varepsilon_2 n))$ are released, if an analyst has prior $p(\mu, \sigma^2) \propto 1$, then their posterior $p(\mu, \sigma^2 \mid \bar{Y}^*, S^{2*})$ is a proper probability distribution.
\end{thm}

Similarly, we can determine posterior propriety for the  constrained analysis.  Since a constrained uniform distribution has bounded support, it is a proper prior distribution and thus produces a proper posterior distribution. Likewise, it can be verified empirically that the posterior produced by $p(\mu, \sigma^2) \propto (\sigma^2)^{-1}$ also is improper in the constrained analysis. 

We use simulation experiments to assess the properties of the posterior distributions based on priors that use or disregard the constraints on the data values. For each of 10,000 simulated datasets of a given $n$, we create a 95\% HPD interval based on the Gibbs sampler in Appendix \ref{app:full_cond}. Figure \ref{fig:coverage} displays the results; see Appendix \ref{app:compute} for run-time details. Without enforcing constraints, the 95\% credible interval has approximately the nominal 95\% coverage rate  regardless of $n$, and the average interval length decays proportionally to $1/n$.
\footnote{A regression of log average interval length on log $n$ has a slope of $-1$ (with $R^2 > 0.999$), indicating that average interval length decays proportionally to $1/n$. A similar result holds for the RMSE; see Appendix \ref{app:add_sims}.}  When enforcing the constraints, the results are similar to those for the unconstrained analysis when $n$ is large. In these cases, credible intervals for $\mu$ are far from the boundary, so that the constraints have little effect. When $n$ is small, the 95\% credible interval has lower than 95\% coverage rate when the true $\mu$ is close to the boundary and near 100\% coverage when $\mu$ is in the middle of the range.  In fact, it can be shown that the prior $p(\mu, \sigma^2) \propto 1$ over the region where $\mu \in [0,1]$ and $\sigma^2 \in (0,\mu(1-\mu)]$ induces the marginal distribution $\mu \sim \textsf{Beta}(2,2)$. This Beta distribution places more probability density in the center of the distribution than near $0$ or $1$, resulting in the observed over and under coverage, depending on the true value of $\mu$. This suggests that $p(\mu, \sigma^2) \propto 1$  is most appropriate when an analyst believes $\mu$ more likely is near the center of the distribution than the tails. That is, while the prior is ``non-informative'' in the unconstrained analysis, it is not so in the constrained analysis.

Enforcing the constraint, however, does have advantages. For all points in Figure \ref{fig:coverage}, the average interval length for the constrained analysis is less than $0.75$. Meanwhile, the average interval length without constraints is greater than $1$ for $n \leq 50$, indicating that the credible intervals must include values that are not within $[0,1]$. Such intervals often include all of $[0,1]$; these are practically useless. Indeed, the calibrated coverage rate for the unconstrained analysis is of dubious merit, as these credible intervals are inflated by permitting probability mass on infeasible regions.  
Additionally, when we estimate $\mu$ with its posterior mode, the RMSE is 
uniformly lower in the constrained analysis than in the unconstrained; see Figure \ref{fig:coverage_RMSE} in Appendix \ref{app:add_sims}. In practical terms, the incorporation of constraints yields estimates that are more sensible scientifically, have less error, and have tighter interval estimates. 

The results described above use a total $\varepsilon$ of $0.2$. When $\varepsilon$ is larger, the constrained and unconstrained analyses are more similar; see Figure \ref{fig:coverage_2} in Appendix \ref{app:add_sims} for additional simulations with total $\varepsilon=2$. When $\varepsilon$ is larger, the unconstrained analysis no longer offers approximately exact coverage: the 95\% posterior credible intervals cover more than 95\% of the time.

\subsection{Default Priors for Example \ref{ex:lead}} \label{sec:prediction}

\begin{figure}[t]
    \centering
    \includegraphics{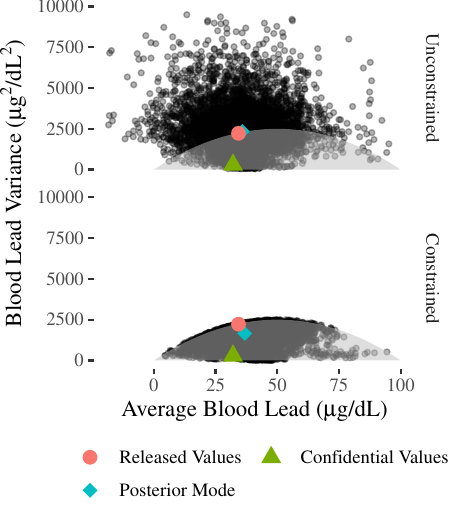}
    \caption{Plot of the joint posterior distribution for $(\mu,\sigma^2)$ in Example \ref{ex:lead} under a uniform prior. The point $(\bar{Y}^* = 34.30, S^2 = 47.17^2)$ is represented by the red circle, the unreleased point $(\bar{Y} = 32.08, S^2 = 16.96^2)$ is represented by the green triangle, and the analyst's posterior mode is represented by the blue diamond. The upper and lower panels provide the posterior when constraints are and are not accounted for, respectively. The shaded area represents the feasible region for $(\mu, \sigma^2)$ from Theorem \ref{thm:par_bounds}. This plot is based on 5,000 Gibbs iterations.}
    \label{fig:bounding_flat}
\end{figure}

Returning to the setting of Example \ref{ex:lead}, we examine the effect of replacing the informative prior used in Section \ref{sec:bounds_lead} with the default prior $p(\mu, \sigma^2) \propto 1$. Figure \ref{fig:bounding_flat} presents a plot analogous to Figure \ref{fig:bounding} with the new prior. In the unconstrained analysis, the posterior mode is approximately equal to the released statistics, but more than 50\% of posterior draws are outside of the feasible region. In the constrained  analysis, as discussed above, the posterior mode for $\mu$ is shifted towards the center of the distribution. A similar effect is observed for $\sigma^2$; the prior has more mass closer to zero and so the posterior mode is smaller than $S^{2*}$.

Analysts also might be interested in using the Bayesian model for prediction. Figure \ref{fig:predictive} plots the posterior predictive distribution for the uniform prior, which is computed via $Y_{\textup{new}(t)} \sim \norm(\mu_{(t)}, \sigma^{2}_{(t)})$ for each Gibbs iteration $t$, under the unconstrained and constrained analyses.\footnote{For the constrained analysis, the distribution is replaced by a truncated Gaussian.} The constrained posterior predictive distribution only includes values in the feasible region of $0$ to $100 \, \mu \textup{g}/\textup{dL}$, with the distribution centered around the posterior mode. The unconstrained posterior predictive distribution, meanwhile, generates 24\% of predictive draws as negative values and 10\% as values greater than $100 \, \mu \textup{g}/\textup{dL}$. This leads to the constrained analysis having substantially lower predictive variance than the unconstrained: the standard deviation is $53 \, \mu \textup{g}/\textup{dL}$ without constraints and $25 \mu \textup{g}/\textup{dL}$ with constraints.

To obtain reasonable predictions with the draws from the unconstrained analysis, an analyst must resort to ad hoc methods. They might, for example, re-code all negative predictions to $0$ and all predictions greater than $100$ to $100$. This would lead to 24\% of predictions being exactly zero, which is not reasonable scientifically, and results in an inflated predictive standard deviation of $35 \, \mu \textup{g}/\textup{dL}$. Alternatively, the analyst could sample predictions from a truncated Gaussian distribution. The variance of this  distribution is still inflated by the samples of overly large values of $\sigma^{2}_{(t)}$, leading to a predictive standard deviation of $27 \, \mu \textup{g}/\textup{dL}$ under this ad hoc approach, which is larger than that of the theoretically principled constrained analysis.

\begin{figure}[t]
    \centering
    \includegraphics {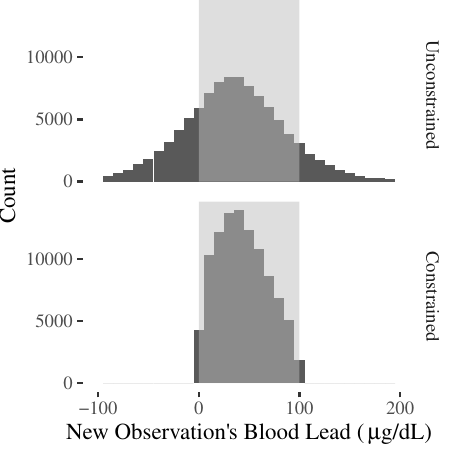}
    \caption{Posterior predictive distribution of a new observation for the posterior draws from Example \ref{ex:lead} and prior $p(\mu, \sigma^2) \propto 1$.
    Upper and lower panels provide the posterior predictive distributions without and without accounting for constraints, respectively. The shaded area represents the feasible region for a new observation. Plots based on 100,000 Gibbs sampler iterations.}
    \label{fig:predictive}
\end{figure}

\section{REGRESSION APPLICATION} \label{sec:regression}

The strategies we discuss can be generalized to any Bayesian analysis under DP. To demonstrate, we consider the DP Bayesian linear regression method of \cite{bernstein2019differentially}. In particular, we adapt an example from \cite{bernstein2019differentially} using data from \cite{data_x20} to illustrate how the strategies for enforcing constraints in Section \ref{sec:bounds} can be used off-the-shelf to supplement an existing method.

\begin{figure*}
    \centering
    \includegraphics{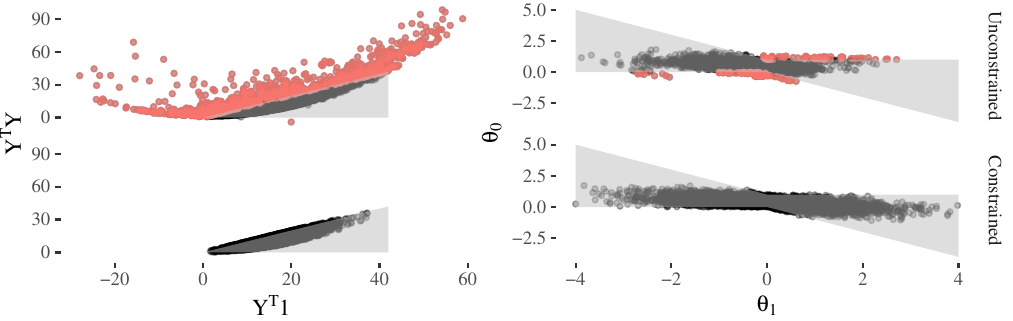}
    \caption{Plot of posterior draws from the linear regression method of \cite{bernstein2019differentially}. The left panels represent draws of the imputed sufficient statistics $\Y^\top \one$ and $\Y^\top \Y$, while the right panels represent draws of the parameters $\theta_1$ and $\theta_0$. The upper and lower panels provide the posterior when constraints are and are not accounted for, respectively. The shaded areas represent the feasible regions; points outside the feasible regions are colored in red. This plot is based on 10,000 Gibbs iterations.}
    \label{fig:BeSh_comparison}
\end{figure*}

The data include the response variable cirrhosis rate, $\{y_i\}$, explanatory variable drinking rate, $\{x_i\}$, and sample size $n = 46$. The authors rescale both variables to lie in $[0,1]$ and assume the range
is public information. Letting $\Y$
be the vector of responses and $\X$
be a design matrix with a column of ones, the authors' model is a simple linear regression of the form $\Y \sim \norm_n(\X\t, \sigma^2 \I_n)$ with the 
conjugate prior
\begin{align}
    \t \mid \sigma^2 \sim \norm_2\left(
    \begin{bmatrix}
        1 \\ 0
    \end{bmatrix},
    \frac{\sigma^2}{4} \mathbf{I}_2
    \right), \quad
    \sigma^2 \sim \textsf{IG}\left(\frac{40}{2}, \frac{1}{2} \right),
\end{align}
where $\mathbf{I}_p$ is the $p \times p$ identity matrix. Their proposed method involves applying the Laplace Mechanism to each element of the sufficient statistics $\X^\top \X$, $\X^\top \Y$ and $\Y^\top \Y$. To ensure a valid solution exists, they enforce that the matrix $\bZ^\top \bZ$, where $\bZ = [\X \hspace{0.08in} \Y]$, is positive semi-definite (PSD) by projecting samples to the nearest PSD matrix. They apply their \texttt{Gibbs-SS-Noisy} method, which also adds Laplace noise to the matrix of fourth sample moments of $\X$. 

Since it is known that $x_i,y_i \in [0,1]$ for all $i$, we can exploit the structure of the linear regression to constrain parameter updates in the Gibbs sampler. In particular, letting $\X = [\one_n \hspace{0.08in} \x_1]$, we show in Appendix \ref{app:BeSh_bounds} that $0 \leq \x_1^\top \x_1 \leq \x_1^\top \one \leq n$, $0 \leq \Y^\top \Y \leq \Y^\top \one \leq n$, $\x_1^\top \Y \leq \min\{\x_1^\top \one, \Y^\top \one \}$, and $0 \leq \sigma^2 \leq 1/4$.  Letting the regression coefficients be $\t = [\theta_0 \hspace{0.08in} \theta_1]^\top$, we know that $(\theta_1, \theta_0)$ must lie in the region depicted in the right panels of Figure \ref{fig:BeSh_comparison}. We create a constrained Gibbs sampler by drawing each iterate from the unconstrained full conditional distribution and rejecting and resampling whenever all constraints are not satisfied. We also reject and resample if $\bZ^\top \bZ$ is not PSD.

Figure \ref{fig:BeSh_comparison} demonstrates some of the effects of enforcing these constraints for a particular set\footnote{The addition of Laplace noise to the fourth sample moments of $\X$ adds additional uncertainty which is not accounted for in the procedure in \cite{bernstein2019differentially}. Because of this, our rejection sampling approach struggles when the released noisy moments are far from their nonprivate counterparts, leading to the Gibbs sampling procedure being unable to find a draw for the next Gibbs iterate satisfying the constraints and getting ``stuck.'' For demonstrative purposes, we choose a seed where this issue is not encountered.} of released values with $\varepsilon = 0.1$ for each of the 11 queries. The shaded regions in the left panels represent the constraint $0 \leq \Y^\top \Y \leq \Y^\top \one \leq n$.
As indicated by the red points, nearly 20\% of posterior draws without accounting for constraints do not satisfy this inequality. The shaded regions in the right panels represent the feasible region for $\t$.
Less than 5\% of the posterior draws did not satisfy the constraint in $\t$. But notably, when all relevant constraints are enforced, the posterior distribution is more spread out throughout the feasible region. That is, not enforcing the constraints leads artificially to a 
model with too little uncertainty about its parameter estimates and thus underestimated uncertainty around downstream predictions. 

\section{DISCUSSION}

A major advantage of the Bayesian paradigm generally is that it facilitates convenient propagation of uncertainty. Indeed, for these reasons, Bayesian inference has been a focus of recent work in the DP literature, e.g.,  \cite{bernstein2018differentially}, \cite{gong2022exact}, and \cite{ju2022data}.
However, inference from these methods may be 
sensitive to the analyst's prior selection. As far as we are aware, our work is the first to examine the effect of the prior distribution on inference in settings with substantial noise due to DP and constraints on data values.

Our case studies suggest that incorporating constraints in Bayesian DP inference, thereby respecting the actual support of the parameters, can result in potentially more accurate representations of posterior distributions, as well as more accurate point estimates and predictions. 
The use of weakly informative or default prior distributions may be advantageous, although analysts must ensure posterior distributions are not improper, as demonstrated in Section \ref{sec:default}. Finally, these results suggest there is scope for further development of default prior distributions to obtain theoretically valid and accurate posterior distributions with desirable frequentist properties.

In this work, we used the simple setting of 
bounded data modeled as following a  
Gaussian distribution as a running case study. Settings where data can be reasonably approximated by a Gaussian distribution are extremely common in applied statistics, but relatively under-explored in the DP literature. This setting also forms the basis for methods applied to a wide array of common tasks in machine learning, such as A/B testing, hierarchical methods, and regression.
Additionally, settings where confidential data values  are assumed to have an inherent or effective bound are widespread in the DP literature. As far as we are aware, this work is the first to examine the effect of incorporating the bounds as a constraint on a downstream analyst's model. Extension of the ideas in this work to settings beyond univariate Gaussian and linear regression models is a promising avenue for future work.

\subsubsection*{Acknowledgements}
This research was supported by NSF grant SES-2217456.

\bibliography{bib}
\bibliographystyle{iclr2025_conference}

%%%%%%%%%%%%%%%%%%%%%%%%%%%%%%%%%%%%%%%%%%%%%%%%%%%%%%%%%%%%

\newpage
\onecolumn
\appendix
\input{supplement}

\end{document}

%% file: macros.tex
\usepackage{amsmath,amsthm,amsfonts,amssymb,amscd}
\usepackage{lastpage}
\usepackage{enumerate}
\usepackage{fancyhdr}
\usepackage{mathrsfs}
\usepackage{xcolor}
\usepackage{graphicx}
\usepackage{listings}
\usepackage{hyperref}
\usepackage{pdfpages}
\usepackage{xfrac}
\usepackage[linesnumbered,algoruled,boxed,lined]{algorithm2e}
\usepackage{placeins}
\usepackage[title]{appendix}
\usepackage{comment} 
\usepackage{tikz-cd}
\usepackage{listings}

\newtheorem{thm}{Theorem}

\newtheorem{cor}{Corollary}
\newtheorem{defi}{Definition}

\newtheorem{exa}{Example}

\def\Y{{\mathbf Y}}
\def\bZ{{\mathbf Z}}

\def\t{{\boldsymbol \theta}}
\def\X{{\mathbf X}}
\def\norm{{\mathcal N}}
\def\I{{\mathbf I}}
\def\one{{\mathbf 1}}

\def\R{{\mathbb R}}

\def\M{{\mathcal M}}

\def\x{{\mathbf x}}

\newtheorem*{rep@theorem}{\rep@title}
\newcommand{\newreptheorem}[2]{%
\newenvironment{rep#1}[1]{%
 \def\rep@title{#2 \ref{##1}}%
 \begin{rep@theorem}}%
 {\end{rep@theorem}}}
\makeatother

\newreptheorem{theorem}{Theorem}
\newreptheorem{lemma}{Lemma}
\newreptheorem{corollary}{Corollary}

%% file: supplement.tex
\section{RE-SCALING THE DATA} \label{app:scale}

Throughout the document, we assume that the data $Y_1, \ldots, Y_n$ are contained in the interval $[a,b]$ for known, public $a$ and $b$. In this section, we prove that,   without loss of generality, we may consider data on the interval $[0,1]$ by re-scaling the data as follows to obtain $\tilde{Y}_i \in [0,1]$,
\begin{align}
    \tilde{Y}_i = \frac{Y_i - a}{b - a}.
\end{align}
We can convert back to the original scale via the relationship $Y_i = (b-a) \tilde{Y}_i + a$.

We now show that the sufficient statistics released via the Laplace Mechanism on the $[0,1]$ scale are re-scaled versions of the sufficient statistics released via the Laplace Mechanism on the $[a,b]$ scale. Thus no information is lost from scaling before the DP release and re-scaling afterwards. 

\begin{reptheorem}{thm:scale}
    Let $Y_1, \ldots, Y_n \in [a,b]$ and let $\tilde{Y}_i = (Y_i - a)/(b - a) \in [0,1]$. Let $\bar{Y}$ and $S^2$ be 
    the sample mean and variance for $\{Y_i\}$ and let $\tilde{\bar{Y}}$ and $\tilde{S}^2$ be the sample mean and variance for $\{ \tilde{Y}_i\}$. Suppose each statistic is released via the Laplace Mechanism under $\varepsilon$-DP and denote the DP statistics $\bar{Y}^*$, $S^{2*}$, $\tilde{\bar{Y}}^*$, and $\tilde{S}^{2*}$. Then, $\bar{Y}^* \overset{d}{=} (b-a)\tilde{\bar{Y}}^* + a$ and $S^{2*} \overset{d}{=} (b-a)^2 \tilde{S}^{2*}$.
\end{reptheorem}

\begin{proof}
    To begin, note that we may relate the sample means on the two scales as follows
    \begin{align}
        \bar{Y} &= \frac{1}{n}\sum_{i=1}^n Y_i = \frac{1}{n}\sum_{i=1}^n \left[(b-a)\tilde{Y}_i + a\right] = (b-a)\frac{1}{n}\sum_{i=1}^n \tilde{Y}_i + a = (b-a) \tilde{\bar{Y}} + a.
    \end{align}
    Similarly, we may relate the sample variances on the two scales as follows.
    \begin{align}
        S^2 &= \frac{1}{n-1}\sum_{i=1}^n (Y_i - \bar{Y})^2 \\
        &= \frac{1}{n-1}\sum_{i=1}^n ([(b-a) \tilde{Y}_i + a] - [(b-a) \tilde{\bar{Y}} + a])^2 \\
        &= (b-a)^2\frac{1}{n-1}\sum_{i=1}^n (\tilde{Y}_i -\tilde{\bar{Y}})^2 \\
        &= (b-a)^2\tilde{S}^2.
    \end{align}

    We now consider the sensitivities of each of the four statistics. By Lemmas 11 and 12 in \cite{du2020differentially}, the sensitivity of $\bar{Y}$ is $(b-a)/n$ and the sensitivity of $S^2$ is $(b-a)^2/n$. Similarly, the sensitivity of $\tilde{\bar{Y}}$ is $1/n$ and the sensitivity of $\tilde{S}^2$ is $1/n$. Thus, when released under $\varepsilon$-DP via the Laplace Mechanism, the released statistics have distribution
    \begin{align}
        \bar{Y}^* &\sim \textsf{Lap}\left(\bar{Y}, \frac{b-a}{\varepsilon n} \right), \qquad S^{2*} \sim \textsf{Lap}\left(S^2, \frac{(b-a)^2}{\varepsilon n} \right) \\
        \tilde{\bar{Y}}^* &\sim \textsf{Lap}\left(\tilde{\bar{Y}}, \frac{1}{\varepsilon n} \right), \qquad \hspace{0.18in} \tilde{S}^{2*} \sim \textsf{Lap}\left(\tilde{S^2}, \frac{1}{\varepsilon n} \right).
    \end{align}

    Recall that if $X$ has a Laplace distribution with location parameter $m$ and scale parameter $b$, i.e., $X\sim \textsf{Lap}(m,b)$, then if $k > 0$ and $c \in \R$, it follows that $kX + c \sim \textsf{Lap}(km + c, kb)$. By this property, $(b-a)\tilde{\bar{Y}}^* + a$ has location parameter $(b-a)\tilde{\bar{Y}} + a = \bar{Y}$ and scale parameter $(b-a)/(\varepsilon n)$. Similarly, $(b-a)^2 \tilde{S}^{2*}$ has location parameter $(b-a)^2 \tilde{S^2} = S^2$ and scale parameter $(b-a)^2/(\varepsilon n)$. The result follows.
\end{proof}

By this result, it is equivalent to release $\bar{Y}^*$ and $S^{2*}$ on the $[a,b]$ scale directly via the Laplace distribution or to re-scale the data to the $[0,1]$ scale, release $\tilde{\bar{Y}}^*$ and $\tilde{S}^{2*}$, and then use the relationships in Theorem \ref{thm:scale} to convert back to the $[a,b]$ scale.

\section{THE TRUNCATED GAMMA MIXTURE DISTRIBUTION} \label{app:TGM}

This section provides additional details related to the truncated gamma mixture distribution (abbreviated TGM). The TGM distribution has four parameters: $\alpha$ is analogous to the gamma distribution's scale parameter, $\beta$ is analogous to the gamma distribution's rate parameter, $\lambda$ determines how far the TGM distribution is from a gamma distribution, and $\tau$ determines the point of truncation. When $\tau > 0$, the distribution is a mixture of $\textsf{Gam}_{(0,\tau]}(\alpha, \beta - \lambda)$ and $\textsf{Gam}_{(\tau, \infty)}(\alpha, \beta + \lambda)$ with the mixture weights in (\ref{eq:weights}). When $\tau \leq 0$, the distribution is equivalent to $\textsf{Gam}(\alpha, \beta + \lambda)$.

Figure \ref{fig:TGM} compares the probability density of a TGM distribution to the probability density of the gamma distribution with the same $\alpha$ and $\beta$. We see that the shapes are similar, but the rates are different on each side of $\tau$. Notably, the distribution is continuous but is clearly not differentiable at $\tau$.

\begin{figure}[t]
    \centering
    \includegraphics{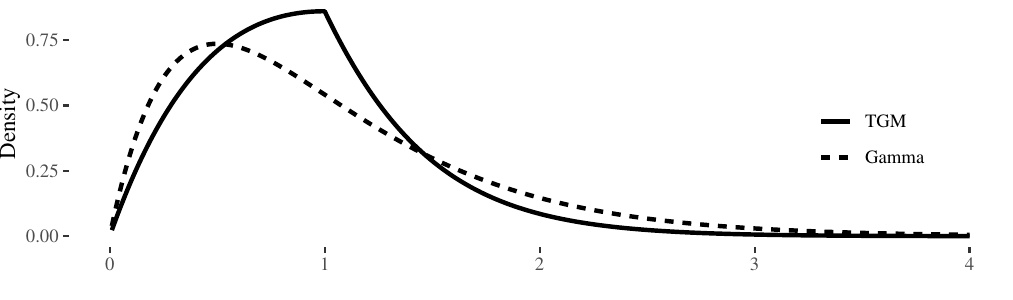}
    \caption{Comparison of the probability density functions of $\textsf{TGM}(\alpha = 2, \beta = 2, \lambda = 1, \tau = 1)$ (the solid line) and $\textsf{Gam}(\alpha = 2, \beta = 2)$ (the dashed line).}
    \label{fig:TGM}
\end{figure}

Algorithm \ref{alg:TGM} provides a straightforward procedure for sampling from the TGM distribution. It assumes one is able to sample from a truncated gamma distribution, which can be done via standard tools such as the \texttt{rtrunc} function from the \texttt{truncdist} package in \texttt{R} \citep[License:~GPL$\geq2$]{truncdist}.

\begin{algorithm}
\label{alg:TGM}
\DontPrintSemicolon
\KwIn{$\alpha$, $\beta$, $\lambda$, $\tau$}
    \uIf{$\tau \leq 0$}{Sample $X \sim \textsf{Gam}(\alpha, \beta + \lambda)$}
    \uElse{Compute $\pi_1$ via (\ref{eq:weights})\;
    Sample $U \sim \textsf{Unif}(0,1)$\;
    \uIf{$U \leq \pi_1$}{
    Sample $X \sim \textsf{Gam}(\alpha, \beta - \lambda)$ truncated to $(0, \tau]$\;
    }
    \uElse{
    Sample $X \sim \textsf{Gam}(\alpha, \beta + \lambda)$ truncated to $(\tau, \infty)$\;
    }}
\KwOut{$X$}
\caption{Sample $X \sim \textsf{TGM}(\alpha, \beta, \lambda, \tau)$}
\end{algorithm}

We now prove that Algorithm \ref{alg:TGM} correctly samples from the TGM distribution.

\begin{thm}\label{thm:TGM_alg}
    $X$ sampled via Algorithm \ref{alg:TGM} has distribution $X \sim \textsf{TGM}(\alpha, \beta, \lambda, \tau)$.
\end{thm}

\begin{proof}
    If $\tau \leq 0$, then $X  \sim \textsf{Gam}(\alpha, \beta + \lambda)$, as desired.

    We now consider the case where $\tau > 0$. Let $F_1(x)$ and $f_1(x)$ be the CDF and PDF, respectively, of $\textsf{Gam}(\alpha, \beta - \lambda)$ truncated to $(0, \tau]$. Let $F_2(x)$ and $f_2(x)$ be the CDF and PDF, respectively, of $\textsf{Gam}(\alpha, \beta + \lambda)$ truncated to $(\tau, \infty)$. By the law of total probability, the CDF of $X$ is
    \begin{align}
        P[X \leq x] &= P[X \leq x \mid U \leq \pi_1] \, P[U \leq \pi_1] + P[X \leq x \mid U > \pi_1] \, P[U > \pi_1] \\
        &= \pi_1 \, F_1(x) + \pi_2 \, F_2(x).
    \end{align}
    This function is differentiable everywhere except $x = \tau$. When $x \neq \tau$, we may take the derivative with respect to $x$ to obtain the PDF
    \begin{align}
        p(x) &= \frac{\textrm{d}}{\textrm{d}x}\left[ \pi_1 \, F_1(x) + \pi_2 \, F_2(x)\right] \\
        &= \pi_1 f_1(x) + \pi_2 f_2(x) \\
        &= \pi_1 \frac{(\beta - \lambda)^\alpha}{\gamma(\alpha, (\beta - \lambda)\tau)} x^{\alpha-1} e^{-(\beta - \lambda)x} \one[x \leq \tau] + \pi_2 \frac{(\beta + \lambda)^\alpha}{\Gamma(\alpha, (\beta + \lambda)\tau)} x^{\alpha-1} e^{-(\beta + \lambda)x} \one[x > \tau].
    \end{align}
    Thus, we have shown that $X \sim \textsf{TGM}(\alpha, \beta, \lambda, \tau)$ for all $X \neq \tau$. Since the event $X = \tau$ occurs with probability zero, this completes the proof.
\end{proof}

The TGM distribution primarily arises in the following setting.

\begin{thm} \label{thm:TGM}
    If $\tau \mid \mu \sim \textsf{Lap}(\mu, 1/\lambda)$ and $\mu \sim \textsf{Gam}(\alpha, \beta)$ for $\beta > \lambda$, then $\mu \mid \tau \sim \textsf{TGM}(\alpha, \beta, \lambda, \tau)$.
\end{thm}

\begin{proof}
    By Bayes' Theorem, the desired distribution is
    \begin{align}
        p(\mu \mid \tau) = \frac{p(\tau \mid \mu) \, p(\mu)}{\int_0^\infty p(\tau \mid \mu) \, p(\mu) \, \textrm{d}\mu}. \label{eq:bayes}
    \end{align}
    We begin by investigating the numerator of (\ref{eq:bayes}), which has the following form.
    \begin{align}
        p(\tau \mid \mu) \, p(\mu) &= \frac{\lambda}{2}e^{-\lambda | \tau - \mu |} \, \frac{\beta^\alpha}{\Gamma(\alpha)} \mu^{\alpha - 1} e^{-\beta \mu} \\
        &= \frac{\lambda \beta^\alpha}{2\Gamma(\alpha)} \mu^{\alpha - 1} \left[ e^{-\lambda\tau} e^{-(\beta - \lambda)\mu} \one[\mu \leq \tau] +  e^{\lambda\tau} e^{-(\beta + \lambda)\mu} \one[\mu > \tau]\right] \\
        &= \frac{\lambda \beta^\alpha}{2\Gamma(\alpha)}\left[e^{-\lambda\tau} \, \mu^{\alpha - 1}  e^{-(\beta - \lambda)\mu} \one[\mu \leq \tau] + e^{\lambda\tau} \, \mu^{\alpha - 1} e^{-(\beta + \lambda)\mu} \one[\mu > \tau] \right] \\
        &= \frac{\lambda \beta^\alpha}{2\Gamma(\alpha)}\bigg[e^{-\lambda\tau} \frac{\gamma(\alpha, (\beta - \lambda)\tau)}{(\beta - \lambda)^\alpha} \, \frac{(\beta - \lambda)^\alpha}{\gamma(\alpha, (\beta - \lambda)\tau)} \mu^{\alpha - 1}  e^{-(\beta - \lambda)\mu} \one[\mu \leq \tau] \nonumber \\
        &\qquad \qquad \quad  + e^{\lambda\tau} \frac{\Gamma(\alpha, (\beta + \lambda)\tau)}{(\beta + \lambda)^\alpha} \, \frac{(\beta + \lambda)^\alpha}{\Gamma(\alpha, (\beta + \lambda)\tau)} \mu^{\alpha - 1} e^{-(\beta + \lambda)\mu} \one[\mu > \tau] \bigg].
    \end{align}
    The denominator of (\ref{eq:bayes}) is then as follows.
    \begin{align}
        \int_0^\infty p(\tau \mid \mu) \, p(\mu) \, \textrm{d}\mu 
        &= \frac{\lambda \beta^\alpha}{2\Gamma(\alpha)}\bigg[e^{-\lambda\tau} \frac{\gamma(\alpha, (\beta - \lambda)\tau)}{(\beta - \lambda)^\alpha} \int_0^\tau \frac{(\beta - \lambda)^\alpha}{\gamma(\alpha, (\beta - \lambda)\tau)} \mu^{\alpha - 1}  e^{-(\beta - \lambda)\mu} \, \textrm{d}\mu \nonumber \\
        &\qquad \qquad \quad  + e^{\lambda\tau} \frac{\Gamma(\alpha, (\beta + \lambda)\tau)}{(\beta + \lambda)^\alpha} \int_\tau^\infty  \frac{(\beta + \lambda)^\alpha}{\Gamma(\alpha, (\beta + \lambda)\tau)} \mu^{\alpha - 1} e^{-(\beta + \lambda)\mu} \, \textrm{d}\mu \bigg] \\
        &= \frac{\lambda \beta^\alpha}{2\Gamma(\alpha)}\bigg[e^{-\lambda\tau} \frac{\gamma(\alpha, (\beta - \lambda)\tau)}{(\beta - \lambda)^\alpha}  + e^{\lambda\tau} \frac{\Gamma(\alpha, (\beta + \lambda)\tau)}{(\beta + \lambda)^\alpha} \bigg].
    \end{align}
    The second equality follows by recognizing the integrands as the PDFs of truncated gamma distributions. Thus, the distribution of $\mu \mid \tau$ is
    \begin{align}
        p(\mu \mid \tau) &= \frac{e^{-\lambda\tau} \frac{\gamma(\alpha, (\beta - \lambda)\tau)}{(\beta - \lambda)^\alpha}}{e^{-\lambda\tau} \frac{\gamma(\alpha, (\beta - \lambda)\tau)}{(\beta - \lambda)^\alpha}  + e^{\lambda\tau} \frac{\Gamma(\alpha, (\beta + \lambda)\tau)}{(\beta + \lambda)^\alpha}} \frac{(\beta - \lambda)^\alpha}{\gamma(\alpha, (\beta - \lambda)\tau)} \mu^{\alpha - 1}  e^{-(\beta - \lambda)\mu} \one[\mu \leq \tau] \nonumber \\
        &\qquad  + \frac{e^{\lambda\tau} \frac{\Gamma(\alpha, (\beta + \lambda)\tau)}{(\beta + \lambda)^\alpha}}{e^{-\lambda\tau} \frac{\gamma(\alpha, (\beta - \lambda)\tau)}{(\beta - \lambda)^\alpha}  + e^{\lambda\tau} \frac{\Gamma(\alpha, (\beta + \lambda)\tau)}{(\beta + \lambda)^\alpha}} \frac{(\beta + \lambda)^\alpha}{\Gamma(\alpha, (\beta + \lambda)\tau)} \mu^{\alpha - 1} e^{-(\beta + \lambda)\mu} \one[\mu > \tau],
    \end{align}
    which we recognize as the PDF of a $\textsf{TGM}(\alpha, \beta, \lambda, \tau)$.
\end{proof}

\section{DERIVATION OF FULL CONDITIONALS FOR UNIVARIATE GAUSSIAN GIBBS SAMPLER} \label{app:Gibbs}

In this section, we derive the full conditionals for the Gibbs sampler in Section \ref{sec:gibbs}. Recall that if $Y_1, \ldots, Y_n \overset{iid}{\sim} \norm(\mu, \sigma^2)$, then the distributions of the sufficient statistics are
\begin{align}
    \bar{Y} \mid \mu, \sigma^2 \sim \norm\left(\mu, \frac{\sigma^2}{n} \right), \qquad S^2 \mid \sigma^2 \sim \textsf{Gam}\left(\frac{n-1}{2}, \frac{n-1}{2\sigma^2} \right),
\end{align}
where $(\bar{Y} \perp\!\!\!\perp S^2) \mid \mu, \sigma^2$. Assuming that the $Y_i \in [0,1]$, then the statistics released via the Laplace Mechanism have distributions
\begin{align}
    \bar{Y}^* \mid \bar{Y} \sim \textsf{Lap}\left(\bar{Y}, \frac{1}{\varepsilon_1 n}\right), \qquad S^{2*} \mid S^2 \sim \textsf{Lap}\left(S^2, \frac{1}{\varepsilon_2 n}\right).
\end{align}
The release mechanism is such that $(\bar{Y}^* \perp\!\!\!\perp S^{2*}) \mid \bar{Y}, S^2$. It is convenient to replace the above distribution of $\bar{Y}^*$ with the following equivalent formulation, as proposed by \cite{park2008bayesian} and used in \cite{bernstein2018differentially, bernstein2019differentially}.
\begin{align}
    \bar{Y}^* \mid \bar{Y}, \omega^2 \sim \norm\left(\bar{Y}, \omega^2 \right), \qquad \omega^2 \sim \textsf{Exp}\left(\frac{\varepsilon_1^2 n^2}{2} \right).
\end{align}
Finally, we take the conjugate prior from the public setting
\begin{align}
    \sigma^2 \sim \textsf{IG}\left(\frac{\nu_0}{2}, \frac{\nu_0\sigma_0^2}{2} \right), \qquad \mu \mid \sigma^2 \sim \norm\left(\mu_0, \frac{\sigma^2}{\kappa_0} \right).
\end{align}

\begin{figure}
    \centering
    \begin{tikzcd}
        &\mu \arrow[d] & \sigma^2 \arrow[d] \arrow[dl] & \textit{Parameters} \\
        \omega^2 \arrow[dr] & \bar{Y} \arrow[d] & S^2 \arrow[d] & \textit{Augmented Variables} \\
        & \bar{Y}^* & S^{2*} & \textit{Released DP Statistics}
    \end{tikzcd}
    \caption{A graphical model representing the structure of the univariate Gaussian setting.}
    \label{fig:graph}
\end{figure}

Figure \ref{fig:graph} presents the graphical model corresponding to this likelihood, which we can factor as follows.
\begin{align}
    p(\bar{Y}^*, S^{2*}, \bar{Y}, S^2, \omega^2 \mid \mu, \sigma^2) &= p(\bar{Y}^* \mid \bar{Y}, \omega^2) \, p(S^{2*} \mid S^2) \, p(\bar{Y} \mid \mu, \sigma^2) \, p(S^2 \mid \sigma^2) \, p(\omega^2) \\
    &= \norm(\bar{Y}^*; \bar{Y}, \omega^2) \, \textsf{Lap}\left(S^{2*}; S^2, \frac{1}{\varepsilon_2 n}\right) \, \norm\left(\bar{Y}; \mu, \frac{\sigma^2}{n}\right) \nonumber \\
    &\qquad \qquad \textsf{Gam}\left(S^2; \frac{n-1}{2}, \frac{n-1}{2\sigma^2}\right) \, \textsf{Exp}\left(\omega^2; \frac{\varepsilon_1^2 n^2}{2} \right)
\end{align}

We now examine each full conditional. The full conditional for $\mu$ is
\begin{align}
    p(\mu \mid \bar{Y}^*, S^{2*}, \bar{Y}, S^2, \omega^2, \sigma^2) &\propto p(\bar{Y} \mid \mu, \sigma^2) \, p(\mu \mid \sigma^2) = \norm\left(\bar{Y}; \mu, \frac{\sigma^2}{n}\right) \, \norm\left(\mu ; \mu_0, \frac{\sigma^2}{\kappa_0}\right),
\end{align}
which we recognize as the full conditional from the public setting, as derived in Chapter 5 of \cite{hoff2009first}. Similarly, the full conditional for $\sigma^2$ is
\begin{align}
    p(\sigma^2 \mid \bar{Y}^*, S^{2*}, \bar{Y}, S^2, \omega^2, \mu) &\propto p(\bar{Y} \mid \mu, \sigma^2) \, p(S^2 \mid \sigma^2) \, p(\sigma^2) \\
    &= \norm\left(\bar{Y}; \mu, \frac{\sigma^2}{n}\right)\, \textsf{Gam}\left(S^2; \frac{n-1}{2}, \frac{n-1}{2\sigma^2}\right) \, \textsf{IG}\left( \frac{\nu_0}{2}, \frac{\nu_0\sigma_0^2}{2} \right).
\end{align}
This is the same full conditional as in the public setting, as derived in Chapter 6 of \cite{hoff2009first}. Thus, the full conditionals have the form
\begin{align}
    \mu \mid \sigma^2, \omega^2, \bar{Y}, S^2, \bar{Y}^*, S^{2*} &\sim \norm\left(\frac{n\bar{Y} + \kappa_0\mu_0}{n + \kappa_0}, \frac{\sigma^2}{n +\kappa_0} \right) \\ 
    \sigma^2 \mid \mu, \omega^2, \bar{Y}, S^2, \bar{Y}^*, S^{2*} &\sim \textsf{IG}\left(\frac{n+\nu_0}{2}, \frac{\nu_0\sigma_0^2 + (n-1)S^2 + n(\bar{Y} - \mu)^2}{2} \right).
\end{align}
We next examine the full conditional for $\bar{Y}$, which is 
\begin{align}
    p(\bar{Y} \mid \bar{Y}^*, S^{2*}, S^2, \omega^2, \mu, \sigma^2) &\propto p(\bar{Y}^* \mid \bar{Y}, \omega^2) \, p(\bar{Y} \mid \mu, \sigma^2) = \norm(\bar{Y}^*; \bar{Y}, \omega^2) \, \norm\left(\bar{Y}; \mu, \frac{\sigma^2}{n}\right).
\end{align}
We recognize this is equivalent to a Gaussian model with known variance and a Gaussian prior on the mean. By Chapter 5 of \cite{hoff2009first}, the full conditional for $\bar{Y}$ is
\begin{align}
    \bar{Y} \mid \mu, \sigma^2, \omega^2, S^2, \bar{Y}^*, S^{2*} &\sim \norm\left(\frac{\frac{\bar{Y}^*}{\omega^2}+ \frac{n\mu}{\sigma^2}}{\frac{1}{\omega^2}+ \frac{n}{\sigma^2}}, \frac{1}{\frac{1}{\omega^2}+ \frac{n}{\sigma^2}} \right).
\end{align}
The full conditional for $\omega^2$ is as follows, expressed in terms of its inverse, $1/\omega^2$.
\begin{align}
    p(1/\omega^2 \mid \bar{Y}^*, S^{2*}, \bar{Y}, S^2, \mu, \sigma^2) &\propto p(\bar{Y}^* \mid \bar{Y}, 1/\omega^2) \, p(1/\omega^2) \\
    &= \norm(\bar{Y}^*; \bar{Y}, (1/\omega^2)^{-1}) \, \textsf{IG}\left(1/\omega^2; 1, \frac{\varepsilon_1^2 n^2}{2} \right).
\end{align}
We recognize this form from the Bayesian LASSO \citep{park2008bayesian}, which yields the distribution
\begin{align}
    1/\omega^2 \mid \mu, \sigma^2, \bar{Y}, S^2, \bar{Y}^*, S^{2*} &\sim \textsf{InvGaus}\left(\frac{\varepsilon_1 n}{|\bar{Y}^* - \bar{Y}|}, \varepsilon_1^2 n^2 \right).
\end{align}
Finally, the full conditional for $S^2$ is
\begin{align}
    p(S^2 \mid \bar{Y}^*, S^{2*}, \bar{Y}, \omega^2, \mu, \sigma^2) &\propto p(S^{2*} \mid S^2) \, p(S^2 \mid \sigma^2) \\
    &= \textsf{Lap}\left(S^{2*}; S^2, \frac{1}{\varepsilon_2 n}\right) \, \textsf{Gam}\left(S^2; \frac{n-1}{2}, \frac{n-1}{2\sigma^2}\right).
\end{align}
By Theorem \ref{thm:TGM} when $(n-1)/(2\sigma^2) > \varepsilon_2 n$, the distribution is
\begin{align}
    S^2 \mid \mu, \sigma^2, \omega^2, \bar{Y}, \bar{Y}^*, S^{2*} &\sim \textsf{TGM}\left(\frac{n-1}{2}, \frac{n-1}{2\sigma^2}, n\varepsilon_2, S^{2*} \right). \label{eq:S_sq}
\end{align}
Note that by Corollary \ref{cor:sigma_sq} of the main text, which is repeated in Section \ref{app:bound} below, since $Y_i \in [0,1]$ it follows that $\sigma^2 \leq 1/4$. Thus, we may use this full conditional when
\begin{align}
    \varepsilon_2 < \frac{n-1}{n} \cdot \frac{1}{2\sigma^2} \leq 2 \, \frac{n-1}{n}.
\end{align}
If $\varepsilon_2 > 2(n-1)/n$, then the full conditional for $S^2$ is not necessarily a TGM without additional assumptions. If reasonable for the application, the analyst could impose the additional constraint that $\sigma^2 < (n-1)/(n\varepsilon_2)$, which is sufficient to ensure the full conditional for $S^2$ is of the form in (\ref{eq:S_sq}). Alternatively, if $n$ is large, an analyst could use an approximation of the full conditional for $S^2$, such as the proposals of \cite{bernstein2018differentially, bernstein2019differentially}.

\section{ALGORITHM AND FULL CONDITIONALS} \label{app:full_cond}

In this section, we provide an explicit algorithm for implementing our Gibbs sampler and enumerate all full conditionals examined in this work. The algorithm is presented below. In the algorithm, we denote the $t^{\textrm{th}}$ draw from the posterior for a variable $x$ as $x_{(t)}$.

\begin{algorithm}
\label{alg:Gibbs}
\DontPrintSemicolon
\KwIn{Data $(\bar{Y}^*, S^{2*})$, privacy budget $(\varepsilon_1, \varepsilon_2)$, sample size $n$, iterations $T$, hyperparameters $(\mu_0, \sigma_0^2, \kappa_0, \nu_0)$}
Initialize $\sigma^2_{(0)}$, $\bar{Y}_{(0)}$, $S^2_{(0)}$, and $\omega^2_{(0)}$\;
\For{$t$ in $1:T$}{
    Sample $\mu_{(t)}$ from full conditional $(\mu \mid \sigma^2 = \sigma^{2}_{(t-1)}, \omega^2 = \omega^{2}_{(t-1)}, \bar{Y} =  \bar{Y}_{(t-1)}, S^2 = S^{2}_{(t-1)}, \bar{Y}^*, S^{2*})$\;
    Sample $\sigma^2_{(t)}$ from full conditional $(\sigma^2 \mid \mu = \mu_{(t)}, \omega^2 = \omega^{2}_{(t-1)}, \bar{Y} =  \bar{Y}_{(t-1)}, S^2 = S^{2}_{(t-1)}, \bar{Y}^*, S^{2*})$\;
    Sample $\bar{Y}_{(t)}$ from full conditional $(\bar{Y} \mid \mu = \mu_{(t)}, \sigma^2 = \sigma^2_{(t)}, \omega^2 = \omega^{2}_{(t-1)}, S^2 = S^{2}_{(t-1)}, \bar{Y}^*, S^{2*})$\;
    Sample $S^2_{(t)}$ from full conditional $(S^2 \mid \mu = \mu_{(t)}, \sigma^2 = \sigma^2_{(t)}, \omega^2 = \omega^{2}_{(t-1)}, \bar{Y} = \bar{Y}_{(t)}, \bar{Y}^*, S^{2*})$\;
    Sample $\omega^2_{(t)}$ from full conditional $(\omega^2 \mid \mu = \mu_{(t)}, \sigma^2 = \sigma^2_{(t)}, \bar{Y} = \bar{Y}_{(t)}, S^2 = S^2_{(t)}, \bar{Y}^*, S^{2*})$\;
}
\KwOut{$(\mu_{(1)}, \ldots, \mu_{(T)})$, $(\sigma^2_{(1)}, \ldots, \sigma^2_{(T)})$}
\caption{Run our proposed Gibbs sampler.}
\end{algorithm}

In Algorithm \ref{alg:Gibbs}, the first step is to initialize values for each of the variables, which we describe further in the following paragraph. The analyst then repeatedly samples from the full conditionals for each of the variables. The forms of the full conditionals under different assumptions are described below. The final step is to output the posterior draws for the two parameters.

The posterior distribution is not very sensitive to the initial values, but well-thought-out values may reduce the runtime. In our examples, we set $\sigma^2_{(0)} = S^2_{(0)} = S^{2*}$ if $S^{2*} \in (0, 1/4)$ and $\bar{Y}_{(0)} = \bar{Y}^*$ if $\bar{Y}^* \in [0,1]$. If $\bar{Y}^*$ or $S^{2*}$ are not in these regions, they are set near the closest boundary (although note that $\sigma^2_{(0)}$ and $S^2_{(0)}$ should not be set exactly equal to zero). We set $\omega^2_{(0)} = 2/(\varepsilon_1^2 n^2)$.

We now summarize the full conditional under different prior distributions for the DP univariate Gaussian setting with likelihood given by $Y_i \overset{iid}{\sim} \norm_{[0,1]}(\mu, \sigma^2)$, $\bar{Y}^* \sim \textsf{Lap}(\bar{Y}, 1/(\varepsilon_1 n))$, and $S^{2*} \sim \textsf{Lap}(S^2, 1/(\varepsilon_2 n))$. Subscripts on a distribution denote truncation of the distribution to an interval.

If the prior is $\sigma^2 \sim \textsf{IG}\left(\nu_0/2, \nu_0\sigma_0^2/2 \right)$ and $\mu \mid \sigma^2 \sim \norm\left(\mu_0, \sigma^2/\kappa_0\right)$ with no constraints enforced, then
\begin{align}
    &\mu \mid \sigma^2, \omega^2, \bar{Y}, S^2, \bar{Y}^*, S^{2*} \sim \norm\left(\frac{n\bar{Y} + \kappa_0\mu_0}{n + \kappa_0}, \frac{\sigma^2}{n +\kappa_0} \right) \\ 
    &\sigma^2 \mid \mu, \omega^2, \bar{Y}, S^2, \bar{Y}^*, S^{2*} \sim \textsf{IG}_{\left(0, \frac{n-1}{2n\varepsilon_2}\right)}\left(\frac{n+\nu_0}{2}, \frac{\nu_0\sigma_0^2 + (n-1)S^2 + n(\bar{Y} - \mu)^2}{2} \right) \\
    &\bar{Y} \mid \mu, \sigma^2, \omega^2, S^2, \bar{Y}^*, S^{2*} \sim \norm\left(\frac{\frac{\bar{Y}^*}{\omega^2}+ \frac{n\mu}{\sigma^2}}{\frac{1}{\omega^2}+ \frac{n}{\sigma^2}}, \frac{1}{\frac{1}{\omega^2}+ \frac{n}{\sigma^2}} \right) \\
    &1/\omega^2 \mid \mu, \sigma^2, \bar{Y}, S^2, \bar{Y}^*, S^{2*} \sim \textsf{InvGaus}\left(\frac{\varepsilon_1 n}{|\bar{Y}^* - \bar{Y}|}, \varepsilon_1^2 n^2 \right) \\
    &S^2 \mid \mu, \sigma^2, \omega^2, \bar{Y}, \bar{Y}^*, S^{2*} \sim \textsf{TGM}\left(\frac{n-1}{2}, \frac{n-1}{2\sigma^2}, n\varepsilon_2, S^{2*} \right).
\end{align}

If the prior is $\sigma^2 \sim \textsf{IG}\left(\nu_0/2, \nu_0\sigma_0^2/2 \right)$ and $\mu \mid \sigma^2 \sim \norm\left(\mu_0, \sigma^2/\kappa_0\right)$ with constraints enforced, then
\begin{align}
    &\mu \mid \sigma^2, \omega^2, \bar{Y}, S^2, \bar{Y}^*, S^{2*} \sim \norm_{\left[1/2 - \sqrt{1/4 - \sigma^2},1/2 + \sqrt{1/4 - \sigma^2}\right]}\left(\frac{n\bar{Y} + \kappa_0\mu_0}{n + \kappa_0}, \frac{\sigma^2}{n +\kappa_0} \right) \\
    &\sigma^2 \mid \mu, \omega^2, \bar{Y}, S^2, \bar{Y}^*, S^{2*} \sim \textsf{IG}_{[0,\min\{\mu(1-\mu), \frac{n-1}{2n\varepsilon_2})}\left(\frac{n+\nu_0}{2}, \frac{\nu_0\sigma_0^2 + (n-1)S^2 + n(\bar{Y} - \mu)^2}{2} \right) \\
    &\bar{Y} \mid \mu, \sigma^2, \omega^2, S^2, \bar{Y}^*, S^{2*} \sim \norm_{\left[1/2 - \sqrt{1/4 - \frac{n-1}{n}S^2},1/2 + \sqrt{1/4 - \frac{n-1}{n}S^2}\right]}\left(\frac{\frac{\bar{Y}^*}{\omega^2}+ \frac{n\mu}{\sigma^2}}{\frac{1}{\omega^2}+ \frac{n}{\sigma^2}}, \frac{1}{\frac{1}{\omega^2}+ \frac{n}{\sigma^2}} \right) \\
    &1/\omega^2 \mid \mu, \sigma^2, \bar{Y}, S^2, \bar{Y}^*, S^{2*} \sim \textsf{InvGaus}\left(\frac{\varepsilon_1 n}{|\bar{Y}^* - \bar{Y}|}, \varepsilon_1^2 n^2 \right) \\
    &S^2 \mid \mu, \sigma^2, \omega^2, \bar{Y}, \bar{Y}^*, S^{2*} \sim \textsf{TGM}_{(0,\frac{n}{n-1}\bar{Y}(1-\bar{Y})]}\left(\frac{n-1}{2}, \frac{n-1}{2\sigma^2}, n\varepsilon_2, S^{2*} \right). 
\end{align}

If the prior is $p(\mu, \sigma^2) \propto 1$ with no constraints enforced, then
\begin{align}
    &\mu \mid \sigma^2, \omega^2, \bar{Y}, S^2, \bar{Y}^*, S^{2*} \sim \norm\left(\bar{Y}, \frac{\sigma^2}{n} \right) \\
    &\sigma^2 \mid \mu, \omega^2, \bar{Y}, S^2, \bar{Y}^*, S^{2*} \sim \textsf{IG}_{\left(0, \frac{n-1}{2n\varepsilon_2}\right)}\left(\frac{n-2}{2}, \frac{(n-1)S^2 + n(\bar{Y} - \mu)^2}{2} \right) \\
    &\bar{Y} \mid \mu, \sigma^2, \omega^2, S^2, \bar{Y}^*, S^{2*} \sim \norm\left(\frac{\frac{\bar{Y}^*}{\omega^2}+ \frac{n\mu}{\sigma^2}}{\frac{1}{\omega^2}+ \frac{n}{\sigma^2}}, \frac{1}{\frac{1}{\omega^2}+ \frac{n}{\sigma^2}} \right) \\
    &1/\omega^2 \mid \mu, \sigma^2, \bar{Y}, S^2, \bar{Y}^*, S^{2*} \sim \textsf{InvGaus}\left(\frac{\varepsilon_1 n}{|\bar{Y}^* - \bar{Y}|}, \varepsilon_1^2 n^2 \right) \\
    &S^2 \mid \mu, \sigma^2, \omega^2, \bar{Y}, \bar{Y}^*, S^{2*} \sim \textsf{TGM}\left(\frac{n-1}{2}, \frac{n-1}{2\sigma^2}, n\varepsilon_2, S^{2*} \right).
\end{align}

If the prior is $p(\mu, \sigma^2) \propto 1$ with constraints enforced, then
\begin{align}
    &\mu \mid \sigma^2, \omega^2, \bar{Y}, S^2, \bar{Y}^*, S^{2*} \sim \norm_{\left[1/2 - \sqrt{1/4 - \sigma^2},1/2 + \sqrt{1/4 - \sigma^2}\right]}\left(\bar{Y}, \frac{\sigma^2}{n} \right) \\
    &\sigma^2 \mid \mu, \omega^2, \bar{Y}, S^2, \bar{Y}^*, S^{2*} \sim \textsf{IG}_{[0,\min\{\mu(1-\mu), \frac{n-1}{2n\varepsilon_2})}\left(\frac{n-2}{2}, \frac{(n-1)S^2 + n(\bar{Y} - \mu)^2}{2} \right) \\
    &\bar{Y} \mid \mu, \sigma^2, \omega^2, S^2, \bar{Y}^*, S^{2*} \sim \norm_{\left[1/2 - \sqrt{1/4 - \frac{n-1}{n}S^2},1/2 + \sqrt{1/4 - \frac{n-1}{n}S^2}\right]}\left(\frac{\frac{\bar{Y}^*}{\omega^2}+ \frac{n\mu}{\sigma^2}}{\frac{1}{\omega^2}+ \frac{n}{\sigma^2}}, \frac{1}{\frac{1}{\omega^2}+ \frac{n}{\sigma^2}} \right) \\
    &1/\omega^2 \mid \mu, \sigma^2, \bar{Y}, S^2, \bar{Y}^*, S^{2*} \sim \textsf{InvGaus}\left(\frac{\varepsilon_1 n}{|\bar{Y}^* - \bar{Y}|}, \varepsilon_1^2 n^2 \right) \\
    &S^2 \mid \mu, \sigma^2, \omega^2, \bar{Y}, \bar{Y}^*, S^{2*} \sim \textsf{TGM}_{(0,\frac{n}{n-1}\bar{Y}(1-\bar{Y})]}\left(\frac{n-1}{2}, \frac{n-1}{2\sigma^2}, n\varepsilon_2, S^{2*} \right).
\end{align}

\section{CONSTRAINTS FOR THE UNIVARIATE GAUSSIAN SETTING} \label{app:bound}

In this section, we provide proofs for results in Section \ref{sec:bounds} of the main text regarding constraints on parameters and statistics in the univariate Gaussian setting. The result regarding constraints on the parameters $\mu$ and $\sigma^2$ is restated below.

\begin{reptheorem}{thm:par_bounds}
    Let $Y_i \in [0,1]$ have moments $E[Y_i] = \mu$ and $V[Y_i] = \sigma^2$.
    It follows that $\sigma^2 \in [0, \mu(1-\mu)]$ and $\mu \in [1/2 - \sqrt{1/4 - \sigma^2}, 1/2 + \sqrt{1/4 - \sigma^2}]$.
\end{reptheorem}

\begin{proof}
    Since $Y_i \in [0,1]$, it follows that $Y_i^2 \leq Y_i$. and so by monotonicity of expectation, 
    \begin{align} \label{eq:sigma_mu}
        \sigma^2 = Var[Y_i] = E[Y_i^2] - E[Y_i]^2 \leq E[Y_i] - E[Y_i]^2 = \mu-\mu^2.
    \end{align}
    Thus, since $Y_i$ must have non-negative variance, if $\mu$ is known then $\sigma^2 \in [0, \mu(1-\mu)]$. On the other hand, if $\sigma^2$ is known then by (\ref{eq:sigma_mu}), $\mu$ must be such that $\mu^2 - \mu + \sigma^2 \leq 0$. Applying the quadratic formula, we find that $\mu^2 - \mu + \sigma^2$ has roots $\mu = 1/2 \pm \sqrt{1/4 - \sigma^2}$. Recognizing that $\mu^2 - \mu + \sigma^2 \leq 0$ when $\mu$ is between these roots yields the desired result.
\end{proof}

The corollary of Theorem \ref{thm:par_bounds} is restated below.

\begin{repcorollary}{cor:sigma_sq}
    If $Y_i \in [0,1]$, then $V[Y_i] = \sigma^2 \leq 1/4$.
\end{repcorollary}

\begin{proof}
    By Theorem \ref{thm:par_bounds} $\sigma^2 \leq \mu(1-\mu)$ and since $f(\mu) = \mu(1-\mu)$ attains its maximum value at $\mu = 1/2$, it follows that $f(\mu) \leq f(1/2) = 1/4$ for all $\mu$. The result follows.
\end{proof}

The result regarding constraints on the sufficient statistics $\bar{Y}$ and $S^2$ is restated below.

\begin{reptheorem}{thm:stat_bounds}
    Let $Y_1, \ldots, Y_n$ be such that each $Y_i \in [0,1]$. Then $S^2 \in [0, n/(n-1)\bar{Y}(1-\bar{Y})]$ and $\bar{Y} \in [1/2 - \sqrt{1/4 - (n-1)/n \cdot S^2}, 1/2 + \sqrt{1/4 - (n-1)/n \cdot S^2}]$.
\end{reptheorem}

\begin{proof}
    We begin by expressing $S^2$ in a more convenient form, as follows.
    \begin{align}
        (n-1)S^2 = \sum_{i=1}^n (Y_i - \bar{Y})^2 = \sum_{i=1}^n Y_i^2 + n\bar{Y}^2 - 2\bar{Y}\sum_{i=1}^n Y_i = \sum_{i=1}^n Y_i^2 - n \bar{Y}^2.
    \end{align}
    Since $Y_i \in [0,1]$, we have that $Y_i \leq Y_i^2$ and so
    \begin{align}
        (n-1)S^2 = \sum_{i=1}^n Y_i^2 - n \bar{Y}^2 \leq \sum_{i=1}^n Y_i - n\bar{Y}^2 = n\bar{Y}(1 - \bar{Y}).
    \end{align}
    Since $(n-1)S^2$ is a sum of squares, it is non-negative and so if $\bar{Y}$ is known, $S^2 \in [0, n/(n-1)\bar{Y}(1-\bar{Y})]$. On the other hand, if $S^2$ is known then $\bar{Y}$ must be such that $\bar{Y}^2 - \bar{Y} + (n-1)/n \cdot S^2 \leq 0$. By an analogous argument to the proof of Theorem \ref{thm:par_bounds}, it follows that $\bar{Y} \in [1/2 - \sqrt{1/4 - (n-1)/n \cdot S^2}, 1/2 + \sqrt{1/4 - (n-1)/n \cdot S^2}]$.
\end{proof}

\section{DETERMINING WHETHER THE POSTERIOR DISTRIBUTION IS PROPER} \label{app:proper}

In this section, we prove that the posterior distribution in the differentially private univariate Gaussian setting is not a proper probability distribution when $p(\mu, \sigma^2) \propto (\sigma^2)^{-1}$, but is a proper probability distribution when $p(\mu, \sigma^2) \propto 1$. The results are reproduced below.

\begin{reptheorem}{thm:IJ_proper}
    For confidential data $Y_1, \ldots, Y_n \overset{iid}{\sim} \norm(\mu, \sigma^2)$ where $\bar{Y}^* \sim \textsf{Lap}(\bar{Y}, 1/(\varepsilon_1 n))$ and $S^{2*} \sim \textsf{Lap}(S^2, 1/(\varepsilon_2 n))$ are released, if an analyst has prior $p(\mu, \sigma^2) \propto (\sigma^2)^{-1}$, then their posterior $p(\mu, \sigma^2 \mid \bar{Y}^*, S^{2*})$ is not a proper probability distribution.
\end{reptheorem}

\begin{proof}
    To begin, note that under the prior $p(\mu, \sigma^2) \propto (\sigma^2)^{-1}$, for observed $\bar{Y}^* = \bar{y}^*$ and $S^{2*} = s^{2*}$ the posterior distribution, if proper, should be 
    \begin{align}
        p(\mu, \sigma^2 \mid \bar{y}^*, s^{2*}) &= \frac{p(\bar{y}^*, s^{2*} \mid \mu, \sigma^2) \, p(\mu, \sigma^2)}{\int_0^\infty \int_{-\infty}^\infty p(\bar{y}^*, s^{2*} \mid \mu, \sigma^2) \, p(\mu, \sigma^2) \, \textrm{d}\mu \, \textrm{d}\sigma^2} \\
        &= \frac{p(\bar{y}^*, s^{2*} \mid \mu, \sigma^2) \, (\sigma^2)^{-1}}{\int_0^\infty \int_{-\infty}^\infty p(\bar{y}^*, s^{2*} \mid \mu, \sigma^2)  \, (\sigma^2)^{-1} \, \textrm{d}\mu \, \textrm{d}\sigma^2}. \label{eq:posterior_IJ}
    \end{align}
    The posterior is a proper probability distribution if the integral in the denominator of (\ref{eq:posterior_IJ}) is finite. We now examine this integral. Introducing the latent variables $\bar{Y} = \bar{y}$ and $S^2 = s^2$ and exploiting the conditional independence structure, this integral is equivalent to
    \begin{align}
        &\int_0^\infty  \int_{-\infty}^\infty p(\bar{y}^*, s^{2*} \mid \mu, \sigma^2)  \, (\sigma^2)^{-1} \, \textrm{d}\mu \, \textrm{d}\sigma^2 \\
        &= \int_0^\infty \int_{-\infty}^\infty \int_0^\infty  \int_{-\infty}^\infty p(\bar{y}^*, s^{2*}, \bar{y}, s^2 \mid \mu, \sigma^2)  \, (\sigma^2)^{-1} \, \textrm{d}\bar{y} \, \textrm{d}s^2 \, \textrm{d}\mu \, \textrm{d}\sigma^2 \\
        &= \int_0^\infty (\sigma^2)^{-1} \int_{-\infty}^\infty \int_0^\infty  \int_{-\infty}^\infty p(\bar{y}^* \mid \bar{y}) \, p(s^{2*} \mid s^2) \, p(\bar{y} \mid \mu, \sigma^2) \, p(s^{2} \mid \sigma^2)  \,  \, \textrm{d}\bar{y} \, \textrm{d}s^2 \, \textrm{d}\mu \, \textrm{d}\sigma^2 \\
        &= \int_0^\infty (\sigma^2)^{-1} \left[\int_0^\infty p(s^{2*} \mid s^2) \, p(s^{2} \mid \sigma^2) \, \textrm{d}s^2 \right] \left[\int_{-\infty}^\infty\int_{-\infty}^\infty p(\bar{y}^* \mid \bar{y})  \, p(\bar{y} \mid \mu, \sigma^2) \, \textrm{d}\bar{y} \, \textrm{d}\mu \right] \, \textrm{d}\sigma^2.
    \end{align}
    We examine the double integral in the second bracket. If we exchange the order of integration,
    \begin{align}
        \int_{-\infty}^\infty\int_{-\infty}^\infty p(\bar{y}^* \mid \bar{y})  \, p(\bar{y} \mid \mu, \sigma^2) \, \textrm{d}\mu \, \textrm{d}\bar{y} 
        &= \int_{-\infty}^\infty p(\bar{y}^* \mid \bar{y}) \int_{-\infty}^\infty  \frac{1}{\sqrt{2\pi\sigma^2}}e^{-\frac{(\bar{y} - \mu)^2}{2\sigma^2}} \, \textrm{d}\mu \, \textrm{d}\bar{y},
    \end{align}
    since we recognize the quantity under the inner integral as the density function of $\norm(\bar{y}, \sigma^2)$, the integral is equal to $1$. Thus,
    \begin{align}
        \int_{-\infty}^\infty\int_{-\infty}^\infty p(\bar{y}^* \mid \bar{y})  \, p(\bar{y} \mid \mu, \sigma^2) \, \textrm{d}\mu \, \textrm{d}\bar{y} &= \int_{-\infty}^\infty p(\bar{y}^* \mid \bar{y}) \, \textrm{d}\bar{y} = \int_{-\infty}^\infty \frac{\varepsilon_1 n}{2} e^{- \varepsilon_1 n |\bar{y}^* - \bar{y}| } \, \textrm{d}\bar{y} = 1, \label{eq:inner_int}
    \end{align}
    since we recognize the quantity under the integral as the density function of $\textsf{Lap}(\bar{y}^*, 1/(\varepsilon_1 n))$. By Fubini's Theorem, since this integral is finite and the quantity being integrated is nonnegative, it is equivalent to the quantity in brackets above. Thus, the quantity of interest reduces to
    \begin{align}
        \int_0^\infty & \int_{-\infty}^\infty p(\bar{y}^*, s^{2*} \mid \mu, \sigma^2)  \, (\sigma^2)^{-1} \, \textrm{d}\mu \, \textrm{d}\sigma^2 = \int_0^\infty (\sigma^2)^{-1} \int_0^\infty p(s^{2*} \mid s^2) \, p(s^{2} \mid \sigma^2) \, \textrm{d}s^2 \, \textrm{d}\sigma^2.
    \end{align}
    If $s^{2*} > 0$, then we bound the inner integral below as follows.
    \begin{align}
        \int_0^\infty p(s^{2*} \mid s^2) \, p(s^{2} \mid \sigma^2) \, \textrm{d}s^2 
        &= \frac{\varepsilon_2 n \left(\frac{n-1}{2\sigma^2} \right)^{\frac{n-1}{2}}}{2 \, \Gamma(\frac{n-1}{2})}\int_0^\infty e^{-\varepsilon_2 n |s^{2*} - s^2|} \, (s^2)^{\frac{n-1}{2} - 1} e^{-\frac{n-1}{2\sigma^2}s^2} \, \textrm{d}s^2 \\
        &= \frac{\varepsilon_2 n \left(\frac{n-1}{2\sigma^2} \right)^{\frac{n-1}{2}}}{2 \, \Gamma(\frac{n-1}{2})}\Bigg[e^{-\varepsilon_2 n s^{2*}}\int_0^{s^{2*}} (s^2)^{\frac{n-1}{2} - 1} e^{-(\frac{n-1}{2\sigma^2} - \varepsilon_2 n)s^2} \, \textrm{d}s^2 \nonumber \\
        & \hspace{1.1in}+ e^{\varepsilon_2 n s^{2*}}\int_{s^{2*}}^\infty (s^2)^{\frac{n-1}{2} - 1} e^{-(\frac{n-1}{2\sigma^2} + \varepsilon_2 n)s^2} \, \textrm{d}s^2 \Bigg] \\
        &\geq \frac{\varepsilon_2 n \left(\frac{n-1}{2\sigma^2} \right)^{\frac{n-1}{2}}}{2 \, \Gamma(\frac{n-1}{2})}\Bigg[e^{-\varepsilon_2 n s^{2*}}\int_0^{s^{2*}} (s^2)^{\frac{n-1}{2} - 1} e^{-(\frac{n-1}{2\sigma^2} + \varepsilon_2 n)s^2} \, \textrm{d}s^2 \nonumber \\
        & \hspace{1.1in}+ e^{-\varepsilon_2 n s^{2*}}\int_{s^{2*}}^\infty (s^2)^{\frac{n-1}{2} - 1} e^{-(\frac{n-1}{2\sigma^2} + \varepsilon_2 n)s^2} \, \textrm{d}s^2 \Bigg] \\
        &= \frac{\varepsilon_2 n \left(\frac{n-1}{2\sigma^2} \right)^{\frac{n-1}{2}}}{2 \, \Gamma(\frac{n-1}{2})}e^{-\varepsilon_2 n s^{2*}}\int_0^\infty (s^2)^{\frac{n-1}{2} - 1} e^{-(\frac{n-1}{2\sigma^2} + \varepsilon_2 n)s^2} \, \textrm{d}s^2 \\
        &= \frac{\varepsilon_2 n}{2}\left(\frac{\frac{n-1}{2\sigma^2}}{\frac{n-1}{2\sigma^2} + \varepsilon_2 n}\right)^{\frac{n-1}{2}} \\
        &= \frac{\varepsilon_2 n}{2}\left(\frac{\frac{n-1}{2\varepsilon_2 n}}{\frac{n-1}{2\varepsilon_2 n} + \sigma^2}\right)^{\frac{n-1}{2}}.
    \end{align}
    The inequality follows since $e^{x} \geq e^{-x}$ for $x > 0$ and since $e^{-(a-b)x} \geq e^{-(a+b)x}$ for $a > b > 0$ and $x > 0$. If $s^{2*} \leq 0$, then the same bound holds by an analogous argument (with the integral from $0$ to $s^{2*}$ omitted). Thus, the quantity of interest is bounded as follows.
    \begin{align}
        \int_0^\infty \int_{-\infty}^\infty p(\bar{y}^*, s^{2*} \mid \mu, \sigma^2)  \, (\sigma^2)^{-1} \, \textrm{d}\mu \, \textrm{d}\sigma^2 &\geq \frac{\varepsilon_2 n}{2} \int_0^\infty (\sigma^2)^{-1}\left(\frac{\frac{n-1}{2\varepsilon_2 n}}{\frac{n-1}{2\varepsilon_2 n} + \sigma^2}\right)^{\frac{n-1}{2}} \, \textrm{d}\sigma^2 \\
        &\geq \frac{\varepsilon_2 n}{2} \int_0^1 (\sigma^2)^{-1}\left(\frac{\frac{n-1}{2\varepsilon_2 n}}{\frac{n-1}{2\varepsilon_2 n} + \sigma^2}\right)^{\frac{n-1}{2}} \, \textrm{d}\sigma^2 \\
        &\geq \frac{\varepsilon_2 n}{2} \int_0^1 (\sigma^2)^{-1}\left(\frac{\frac{n-1}{2\varepsilon_2 n}}{\frac{n-1}{2\varepsilon_2 n} + 1}\right)^{\frac{n-1}{2}} \, \textrm{d}\sigma^2 \\
        &= \frac{\varepsilon_2 n}{2} \left(\frac{\frac{n-1}{2\varepsilon_2 n}}{\frac{n-1}{2\varepsilon_2 n} + 1}\right)^{\frac{n-1}{2}} \int_0^1 (\sigma^2)^{-1} \, \textrm{d}\sigma^2.
    \end{align}
    Since this integral diverges, the quantity of interest diverges. Thus, \eqref{eq:posterior_IJ} is not a proper probability distribution.
\end{proof}

\begin{reptheorem}{thm:Unif_proper}
    For confidential data $Y_1, \ldots, Y_n \overset{iid}{\sim} \norm(\mu, \sigma^2)$ where $n > 3$ and $\bar{Y}^* \sim \textsf{Lap}(\bar{Y}, 1/(\varepsilon_1 n))$ and $S^{2*} \sim \textsf{Lap}(S^2, 1/(\varepsilon_2 n))$ are released, if an analyst has prior $p(\mu, \sigma^2) \propto 1$, then their posterior $p(\mu, \sigma^2 \mid \bar{Y}^*, S^{2*})$ is a proper probability distribution.
\end{reptheorem}

\begin{proof}
    To begin, note that under the prior $p(\mu, \sigma^2) \propto 1$, for observed $\bar{Y}^* = \bar{y}^*$ and $S^{2*} = s^{2*}$ the posterior distribution, if proper, should be
    \begin{align}
        p(\mu, \sigma^2 \mid \bar{y}^*, s^{2*}) &= \frac{p(\bar{y}^*, s^{2*} \mid \mu, \sigma^2) \, p(\mu, \sigma^2)}{\int_0^\infty \int_{-\infty}^\infty p(\bar{y}^*, s^{2*} \mid \mu, \sigma^2) \, p(\mu, \sigma^2) \, \textrm{d}\mu \, \textrm{d}\sigma^2} \\
        &= \frac{p(\bar{y}^*, s^{2*} \mid \mu, \sigma^2)}{\int_0^\infty \int_{-\infty}^\infty p(\bar{y}^*, s^{2*} \mid \mu, \sigma^2)  \, \textrm{d}\mu \, \textrm{d}\sigma^2}. \label{eq:posterior_Unif}
    \end{align}
    The posterior is a proper probability distribution if the integral in the denominator of (\ref{eq:posterior_Unif}) is finite. By analogy to the proof of Theorem \ref{thm:IJ_proper}, this integral is equivalent to
    \begin{align}
        \int_0^\infty & \int_{-\infty}^\infty p(\bar{y}^*, s^{2*} \mid \mu, \sigma^2)  \, \textrm{d}\mu \, \textrm{d}\sigma^2 \\
        &= \int_0^\infty\left[\int_0^\infty p(s^{2*} \mid s^2) \, p(s^{2} \mid \sigma^2) \, \textrm{d}s^2 \right] \left[\int_{-\infty}^\infty\int_{-\infty}^\infty p(\bar{y}^* \mid \bar{y})  \, p(\bar{y} \mid \mu, \sigma^2) \, \textrm{d}\bar{y} \, \textrm{d}\mu \right] \, \textrm{d}\sigma^2.
    \end{align}
    By (\ref{eq:inner_int}), the double integral in the second bracket is equal to 1. Thus, the quantity of interest reduces to
    \begin{align}
        \int_0^\infty & \int_{-\infty}^\infty p(\bar{y}^*, s^{2*} \mid \mu, \sigma^2) \, \textrm{d}\mu \, \textrm{d}\sigma^2 = \int_0^\infty \int_0^\infty p(s^{2*} \mid s^2) \, p(s^{2} \mid \sigma^2) \, \textrm{d}s^2 \, \textrm{d}\sigma^2.
    \end{align}
    Exchanging the order of integration yields
    \begin{align}
        \int_0^\infty p(s^{2*} \mid s^2)\int_0^\infty \, p(s^{2} \mid \sigma^2) \, \textrm{d}\sigma^2 \, \textrm{d}s^2.
    \end{align}
    Since $n > 3$, the inner integral is the kernel of an inverse-gamma distribution and so
    \begin{align}
        \int_0^\infty \, p(s^{2} \mid \sigma^2) \, \textrm{d}\sigma^2 &= \int_0^\infty \frac{(\frac{n-1}{2\sigma^2})^{\frac{n-1}{2}}}{\Gamma(\frac{n-1}{2})} (s^2)^{\frac{n-1}{2}-1} e^{-\frac{(n-1)s^2}{2\sigma^2}} \, \textrm{d}\sigma^2 \\
        &= \frac{(\frac{n-1}{2})^{\frac{n-1}{2}}}{\Gamma(\frac{n-1}{2})}(s^2)^{\frac{n-1}{2}-1} \int_0^\infty (\sigma^2)^{-\frac{n-3}{2}-1}  e^{-\frac{(n-1)s^2}{2\sigma^2}} \, \textrm{d}\sigma^2 \\
        &= \frac{(\frac{n-1}{2})^{\frac{n-1}{2}}}{\Gamma(\frac{n-1}{2})}(s^2)^{\frac{n-1}{2}-1} \frac{\Gamma(\frac{n-3}{2})}{(\frac{(n-1)s^2}{2})^{\frac{n-3}{2}}} \\
        &= \frac{\frac{n-1}{2}\Gamma(\frac{n-3}{2})}{\Gamma(\frac{n-1}{2}) } \\
        &= \frac{n-1}{n-3}.
    \end{align}
    Thus,
    \begin{align}
        \int_0^\infty p(s^{2*} \mid s^2)\int_0^\infty \, p(s^{2} \mid \sigma^2) \, \textrm{d}\sigma^2 \, \textrm{d}s^2 &= \frac{n-1}{n-3}\int_0^\infty \frac{\varepsilon_2 n}{2}e^{-\varepsilon_2 n |s^{2*} - s^2|} \, \textrm{d}s^2 \\
        &= \frac{n-1}{n-3} \left(\frac{1}{2} + \frac{1}{2}\textup{sgn}\left(s^{2*}\right)\left(1 - e^{-\varepsilon_2 n |s^{2*}|} \right) \right),
    \end{align}
    since we recognize the quantity under the integral as the density function of $\textsf{Lap}(s^{2*}, 1/(\varepsilon_2 n))$, we plug in the form of its CDF at zero. By Fubini's Theorem, since this integral is finite and the quantity being integrated is nonnegative, it is equivalent to the original quantity of interest. Thus,
    \begin{align}
        \int_0^\infty & \int_{-\infty}^\infty p(\bar{y}^*, s^{2*} \mid \mu, \sigma^2)  \, (\sigma^2)^0 \, \textrm{d}\mu \, \textrm{d}\sigma^2 = \frac{n-1}{n-3} \left(\frac{1}{2} + \frac{1}{2}\textup{sgn}\left(s^{2*}\right)\left(1 - e^{-\varepsilon_2 n |s^{2*}|} \right) \right) < \infty
    \end{align}
    and so the posterior is a proper probability distribution.
\end{proof}

\section{NONPRIVATE LAPLACE-UNIFORM POSTERIOR DISTRIBUTION} \label{app:lap_unif}

In this section, we demonstrate that if $X \sim \textsf{Lap}(\xi, 1/\lambda)$ and $\lambda$ is known, then a uniform prior on $\xi$ yields a proper, frequentist matching posterior. The main result is as follows.

\begin{thm} \label{thm:lap_unif}
    Let $X \sim \textsf{Lap}(\xi, 1/\lambda)$, where $\lambda$ is known. Then the prior $p(\xi) \propto 1$ yields a proper posterior distribution that is frequentist matching.
\end{thm}

\begin{proof}
    Applying properties of the Laplace distribution, if $X \sim \textsf{Lap}(\xi, 1/\lambda)$, then $\xi - X \sim \textsf{Lap}(0, 1/\lambda)$. Note that this is a pivotal quantity, since its distribution does not depend on any unknown parameters. Let $\ell_{\alpha/2}$ and $\ell_{1-\alpha/2}$ be the $\alpha/2$ and $(1-\alpha/2)$ quantiles of $ \textsf{Lap}(\xi, 1/\lambda)$. Then, 
    \begin{align}
       P[X + \ell_{\alpha/2} \leq \xi \leq X + \ell_{1-\alpha/2}] = P[\ell_{\alpha/2} \leq \xi - X \leq \ell_{1-\alpha/2}] = 1-\alpha
    \end{align}
    and so $[X + \ell_{\alpha/2}, X + \ell_{1-\alpha/2}]$ is a $(1-\alpha)$ confidence interval for $\xi$.

    In the Bayesian setting, if $p(\xi) \propto 1$ then
    \begin{align*}
        p(\xi \mid X) \propto p(X \mid \xi) \, p(\xi) \propto \exp\left\{-\lambda | X - \xi | \right\} = \exp\left\{-\lambda | \xi - X | \right\}.
    \end{align*}
    Thus, the posterior distribution is $(\xi \mid X) \sim \textsf{Lap}(X, 1/\lambda)$. Again applying properties of the Laplace distribution, $(\xi - X \mid X) \sim \textsf{Lap}(0, 1/\lambda)$. Thus,
    \begin{align}
       P[X + \ell_{\alpha/2} \leq \xi \leq X + \ell_{1-\alpha/2} \mid X] = P[\ell_{\alpha/2} \leq \xi - X \leq \ell_{1-\alpha/2} \mid X] = 1-\alpha
    \end{align}
    and so $[X + \ell_{\alpha/2}, X + \ell_{1-\alpha/2}]$ is a $(1-\alpha)$ credible interval. Since this interval is the same as the confidence interval above, this posterior distribution is frequentist matching.
\end{proof}

\section{ADDITIONAL SIMULATION STUDIES} \label{app:add_sims}

\begin{figure}
    \centering
    \includegraphics{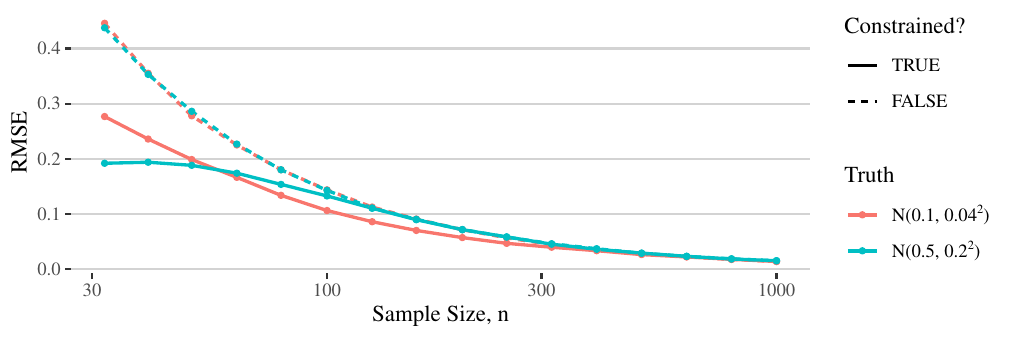}
    \caption{Root mean square error (RMSE) for estimating $\mu$ with the posterior mode for different sample sizes for the simulation in Section 4 of the main text. Results based on 10,000 simulated datasets $Y_i \in [0,1]$ released with $\varepsilon_1 = \varepsilon_2 = 0.1$ and analyzed with prior $p(\mu, \sigma^2) \propto 1$. Analyses with constraints accounted for are solid lines and not accounted for are dashed lines. Data generating model is either $\norm(\mu = 0.1, \sigma^2 = 0.04^2)$ (red) or $\norm(\mu = 0.5, \sigma^2 = 0.2^2)$ (blue). Each Gibbs sampler is run for 20,000 iterations.}
    \label{fig:coverage_RMSE}
\end{figure}

Figure \ref{fig:coverage_RMSE} presents the RMSE from estimating $\mu$ with the posterior mode in the simulation study from Section 4 used to make Figure \ref{fig:coverage} of the main text. The RMSE is uniformly lower in the constrained analysis than in the unconstrained analysis. For the unconstrained analysis, a regression of log RMSE on log $n$ has a slope of $-1$ (with $R^2 > 0.999$), indicating that RMSE decays proportionally to $1/n$

Figure \ref{fig:coverage_2} presents the results of a simulation study identical to that of Figure \ref{fig:coverage} and Figure \ref{fig:coverage_RMSE} of the main text, but with $\varepsilon$ increased by a factor of $10$. There is little practical difference between the constrained and unconstrained analysis at this larger $\varepsilon$. 

\begin{figure}[ht]
    \centering
    \includegraphics{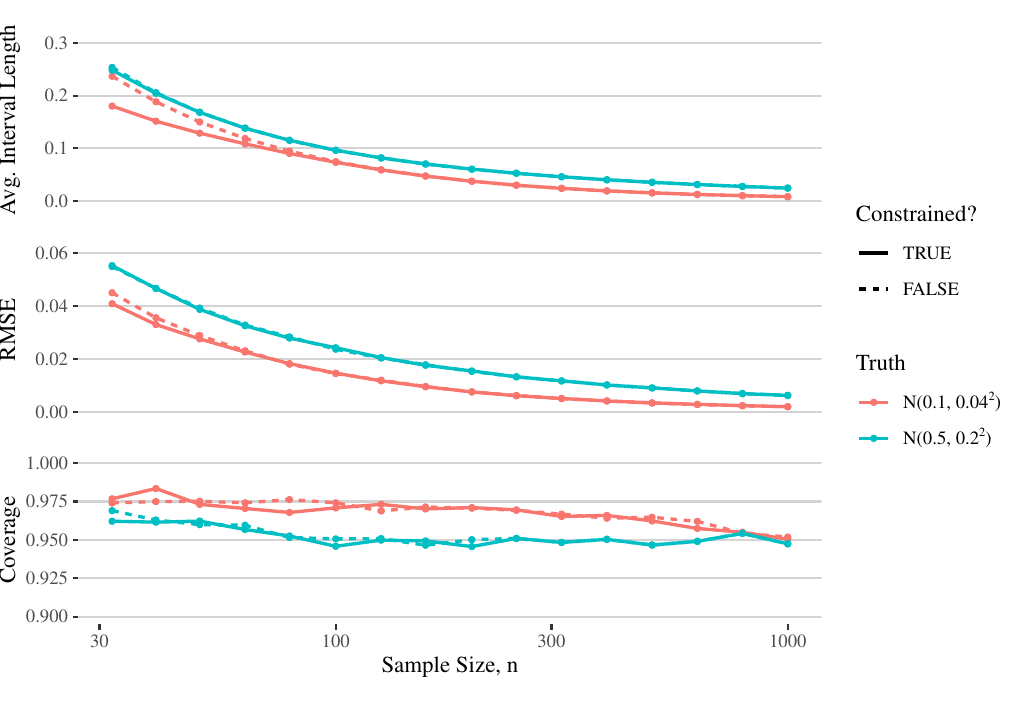}
    \caption{Average 95\% HPD interval length for $\mu$ (top), root mean square error for estimating $\mu$ (middle), and empirical coverage rate of 95\% HPD intervals for $\mu$ (bottom) for different sample sizes. Results based on 10,000 simulated datasets $Y_i \in [0,1]$ released with $\varepsilon_1 = \varepsilon_2 = 1$ and analyzed with prior $p(\mu, \sigma^2) \propto 1$. Analyses with constraints accounted for are solid lines and not accounted for are dashed lines. Data generating model is either $\norm(\mu = 0.1, \sigma^2 = 0.04^2)$ (red) or $\norm(\mu = 0.5, \sigma^2 = 0.2^2)$ (blue). Each Gibbs sampler is run for 20,000 iterations.}
    \label{fig:coverage_2}
\end{figure}

\section{CONSTRAINTS COMPARISON} \label{app:constraints_comp}

In the main text, we focus on accounting for bounds on the data values via constraints on the sample moments (and the model parameters). An alternative method is to enforce the constraint on the data values themselves via the likelihood function instead of via the sample moments
(while still constraining the parameters to the feasible region determined by Theorem \ref{thm:par_bounds}).
In this section, we compare inferences from this alternative method of accounting for the constraints to  inferences from our approach. 

The Gibbs samplers in Section \ref{sec:gibbs}, as well as methods based on sufficient statistics like the samplers of \cite{bernstein2018differentially,bernstein2019differentially}, are unable to easily incorporate constraints on likelihoods. Thus, to make this comparison, we develop and utilize a constraint-aware version of the Gibbs sampler of \cite{ju2022data}. 

\subsection{Enforcing Constraints on the Likelihood} \label{app:Ju_samp}

To begin, we describe the Gibbs sampler proposed by \cite{ju2022data} for the univariate Gaussian setting where  
we do not have constraints on data values. 
At the $t^{\textrm{th}}$ iteration, we sample from $n+2$ full conditionals: those of $\mu$, $\sigma^2$, and each of the confidential $Y_i$. The full conditionals for $\mu$ and $\sigma^2$ are identical to the public setting (and so have the forms in Appendix \ref{app:full_cond}). For each $Y_i$, a proposal is drawn from $Y_i' \sim \norm(\mu_{(t)}, \sigma^2_{(t)})$. Let $\bar{Y}_{\textsf{prev}}$ and $\bar{Y}'$  be the sample mean with and without the proposal, respectively, i.e.,
\begin{align}
    \bar{Y}_{\textsf{prev}} = \frac{1}{n}\left(\sum_{j=1}^{i-1} Y_{i(t)} + Y_{i(t-1)} + \sum_{j=i+1}^n Y_{i(t-1)} \right), \qquad \bar{Y}' = \frac{1}{n}\left(\sum_{j=1}^{i-1} Y_{i(t)} + Y_{i}' + \sum_{j=i+1}^n Y_{i(t-1)} \right).
\end{align}
Let $S^2_{\textsf{prev}}$ and $S^{2\prime}$ be defined similarly for the sample variance. Notably, for a given proposal $Y_i'$, the updated moments $(\bar{Y}', S^{2\prime})$ can be computed from $(\bar{Y}_{\textsf{prev}}, S^2_{\textsf{prev}})$ in constant time. Then, the proposed $Y_i$ is accepted with probability, 
\begin{align}
 r =  \min\left\{1, \exp\left[-\varepsilon_1 n \left(|\bar{Y}^* - \bar{Y}'| - |\bar{Y}^* - \bar{Y}_{\textsf{prev}}|\right) - \varepsilon_2 n \left(|S^{2*} - S^{2\prime}| - |S^{2*} - S^2_{\textsf{prev}}| \right) \right] \right\}.
\end{align}

One can verify empirically that this Gibbs sampler produces a posterior distribution identical to that of our proposed Gibbs sampler without the constraints on the sample moments. Nonetheless, as described in Section \ref{sec:gibbs}, the sampler of \cite{ju2022data} has computational complexity $\mathcal{O}(nT)$, where $T$ is the number of iterations of the Gibbs sampler, whereas our proposed sampler has computational complexity $\mathcal{O}(T)$.

We now describe how to modify the sampler of \cite{ju2022data} to enforce the constraint that each $Y_i \in [a,b]$. Basically, we apply methods similar to those described in Section \ref{sec:bounds}. For $\mu$ and $\sigma^2$, we constrain the full conditionals as in Theorem \ref{thm:par_bounds}.  We also specify a proper prior distribution,  which is a required assumption to utilize the sampler of \cite{ju2022data}. For the illustrations below, we utilize the proper prior $p(\mu, \sigma^2) \propto 1$ over the feasible region of the parameter space.  For each $Y_i$, we enforce the constraint on the likelihood as part of the sampler by rejecting any proposed draws outside of $[a,b]$. 

\subsection{Comparing the Constrained Posterior Distributions}\label{sec:compareconst}

Consider a setting with  sample size $n = 50$ and privacy budget $\varepsilon_1 = \varepsilon_2 = 0.25$  where it is known that $Y_i \in [0,1]$ for all $i$. As a first illustration, we suppose that $\bar{Y}=0.5$ is in the center of the constrained region and $S^2=0.2^2$ is fairly large. Figure \ref{fig:constraints_comp_1} plots the posterior distribution of $(\mu, \sigma^2 \mid \bar{Y}^*, S^{2*})$ for a particular draw of $(\bar{Y}^*, S^{2*})$ from the Laplace mechanism. As shown in the left panel of Figure \ref{fig:constraints_comp_1},  our proposed Gibbs sampler from Section \ref{sec:gibbs} yields a posterior distribution with mode around the released values.  The posterior has considerable uncertainty and places probability density throughout the feasible region. As shown in the right panel of Figure \ref{fig:constraints_comp_1}, the constrained sampler described above in Section \ref{app:Ju_samp} also yields a posterior with its estimate for $\mu$ close to the released $\bar{Y}^*$.  However, its estimate for $\sigma^2$ is considerably smaller than the released $S^{2*}$. The posterior distribution has substantially less uncertainty, concentrating the probability density in the region where $\sigma^2$ is small and $\mu$ is near $\bar{Y}^*$. 

The difference in the two posterior distributions stems from the implications of constructing draws of the sample moments $(\bar{Y}, S^2)$ by individually sampling constrained $Y_i$, as in our adaptation of the sampler of \cite{ju2022data},
or sampling directly from the sample moments' full conditional distributions, as in our proposed approach. 
Recall from Theorem \ref{thm:stat_bounds} that any data bounded in $[a,b]$ has bounded sample variance, $S^2$. Since the modified sampler of \cite{ju2022data} favors values of $Y_i$ that have high density under a constrained normal model, i.e, are far from the bounds $a$ and $b$, the $Y_i$ have a sample variance in the lower part of the feasible range for $S^2$.
As a result, the modified sampler of \cite{ju2022data} favors relatively small values of $\sigma^2$, as observed in Figure \ref{fig:constraints_comp_1}. In contrast, our sampler gives posterior density to values of $S^2$ in the upper part of the feasible range indicated by the bound in Theorem \ref{thm:stat_bounds}, as it is not tied to finding Gaussian models that allow for high density samples of each individual $Y_i$.

\begin{figure}[t]
    \centering
    \includegraphics{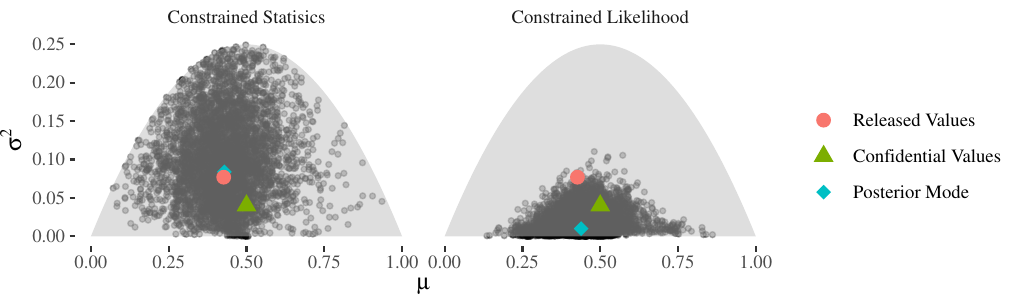}
    \caption{The joint posterior distribution for $(\mu,\sigma^2)$ under a uniform prior with $n = 50$, $\varepsilon_1 = \varepsilon_2 = 0.25$, $a = 0$, and $b = 1$. The unreleased confidential statistics $(\bar{Y} = 0.5, S^2 = 0.2^2)$ are represented by the green triangle, the released statistics $(\bar{Y}^* = 0.43, S^2 = 0.28^2)$ are represented by the red circle, and the analyst's posterior mode is represented by the blue diamond. The left and right panels provide the posterior when the statistics are constrained and when the likelihood is constrained, respectively. The shaded area represents the feasible region for $(\mu, \sigma^2)$ from Theorem \ref{thm:par_bounds}. This plot is based on 25,000 Gibbs iterations, pruned to every 5th draw.}
    \label{fig:constraints_comp_1}
\end{figure}

\begin{figure}[t]
    \centering
    \includegraphics{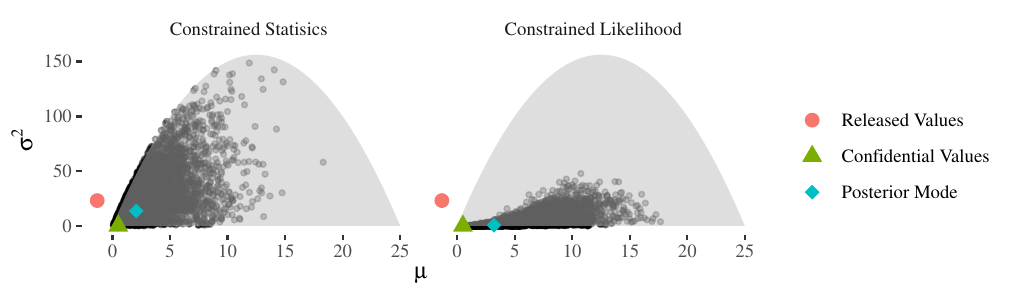}
    \caption{The joint posterior distribution for $(\mu,\sigma^2)$ under a uniform prior with $n = 50$, $\varepsilon_1 = \varepsilon_2 = 0.25$, $a = 0$, and $b = 25$. The unreleased confidential statistics $(\bar{Y} = 0.5, S^2 = 0.2^2)$ are represented by the green triangle, the released statistics $(\bar{Y}^* = -1.3, S^2 = 4.8^2)$ are represented by the red circle, and the analyst's posterior mode is represented by the blue diamond. The left and right panels provide the posterior when the statistics are constrained and when the likelihood is constrained, respectively. The shaded area represents the feasible region for $(\mu, \sigma^2)$ from Theorem \ref{thm:par_bounds}. This plot is based on 25,000 Gibbs iterations, pruned to every 5th draw.}
    \label{fig:constraints_comp_2}
\end{figure}

The tendency of the modified sampler of \cite{ju2022data} to favor values of $(\mu, \sigma^2)$ that generate samples of $Y_i$ relatively far from  $[a, b]$ can have advantages in situations where in fact the data values are far from the bounds.  However, this tendency can be problematic when data are close to the bounds, as we  now illustrate. Consider the same dataset with the same privacy budget, except we suppose the data curator assumes a much wider region for the confidential data: all $Y_i \in [0, 25]$. We keep the observed statistics as $\bar{Y} = 0.5$ and $S^2 = 0.2^2$, which are now very close to the boundary of the feasible region.  

Figure \ref{fig:constraints_comp_2} plots the posterior distributions for the same draw from the Laplace mechanism as in Figure \ref{fig:constraints_comp_1} (with updated scale parameter). Our proposed Gibbs sampler from Section \ref{sec:gibbs} yields a posterior distribution with its mode near the boundary. Since $\bar{Y}^*$ is negative, it focuses the probability density in the part of the feasible region closest to the released value, with substantial uncertainty. 
However, the modified sampler of \cite{ju2022data} 
still yields a posterior distribution with substantial probability density near the center of the feasible region.
The posterior distribution places fairly little probability density in the region where $\mu$ is near zero, despite this being the most likely region for $\bar{Y}$ under the Laplace Mechanism when the released $\bar{Y}^*$ is negative. As discussed above, the need to sample $Y_i$ with high density has the strong effect of favoring $(\mu, \sigma^2)$ values that make it less likely to generate $Y_i$ near the boundaries. Additionally, as discussed in Section \ref{sec:flat_priors}, the constrained uniform prior on $(\mu, \sigma^2)$ induces a marginal prior for $\mu$ with more probability mass in the center of the feasible region, heightening the discrepancy between the observed $(\bar{Y}^*, S^{2*})$ and the posterior inference for $(\mu, \sigma^2)$.

\subsection{Coverage and Error Comparison}

Of course, the previous section considered only two examples to provide intuition. 
We now compare the behaviors under repeated sampling of our proposed sampler from Section \ref{sec:gibbs} and the modified sampler of \cite{ju2022data} described  in Section \ref{app:Ju_samp}. For this comparison, we use the simulation study described in Section \ref{sec:flat_priors}; the results are presented in Figure \ref{fig:coverage_comp}. Once $n$ is large enough (e.g., $n\geq 100$), we see little difference in the average interval lengths or RMSEs of the point estimators for the two methods.  However, until $n$ reaches $500$, the confidence interval coverage rate for the modified sampler of \cite{ju2022data} is notably lower than 95\% when the parameters are near the boundary.  In contrast, the coverage rates for our sampler are nearly 95\% for all values of $n\geq 100$,  with similar average lengths and RMSEs as the modified sampler of \cite{ju2022data}.

When $n<100$, we see sharper differences in the performances of the two methods. 
The biggest difference is with respect to the coverage rates. When the true $\mu$ is near the center, both methods have near 100\% coverage for small $n$. However, when the true $\mu$ is near the boundary, the likelihood constraint produces interval estimates with very poor coverage for small $n$. The coverage is below 70\% for $n < 100$ and near zero when $n$ is very small. Our sampler offers higher coverage rates, although they remain below nominal when $n<100$.  We also see that, for the smaller values of $n$, the modified sampler of \cite{ju2022data} has shorter average interval lengths than our approach, reflecting its tendency to favor small values of $\sigma^2$ as discussed in Section \ref{sec:compareconst}. Additionally, for $\mu$ near the boundary of the parameter space and small $n$, the modified sampler of \cite{ju2022data}  tends to have larger RMSE, reflecting the tendency of the sampler (and the prior) to favor central values of $\mu$ as discussed in Section \ref{sec:compareconst}. This natural tendency is a benefit when $\mu=0.5$, but a deficiency when $\mu = 0.1$.

In practice, of course, an analyst does not know whether $\mu$ is near the center of the region or near the boundary. 
One possibility is to spend some privacy budget to assess which model may be preferred. For example, data curators may provide some differentially private measure of how many observations are close to the bounds, which analysts might be able to use to decide which of the models is likely to be more accurate. Additionally, analysts might  consider utilizing a different default prior distribution, such as one where the marginal distribution for $\mu$ is uniform.
Development of these ideas
is an interesting avenue for future work.

\begin{figure}[t]
    \centering
    \includegraphics{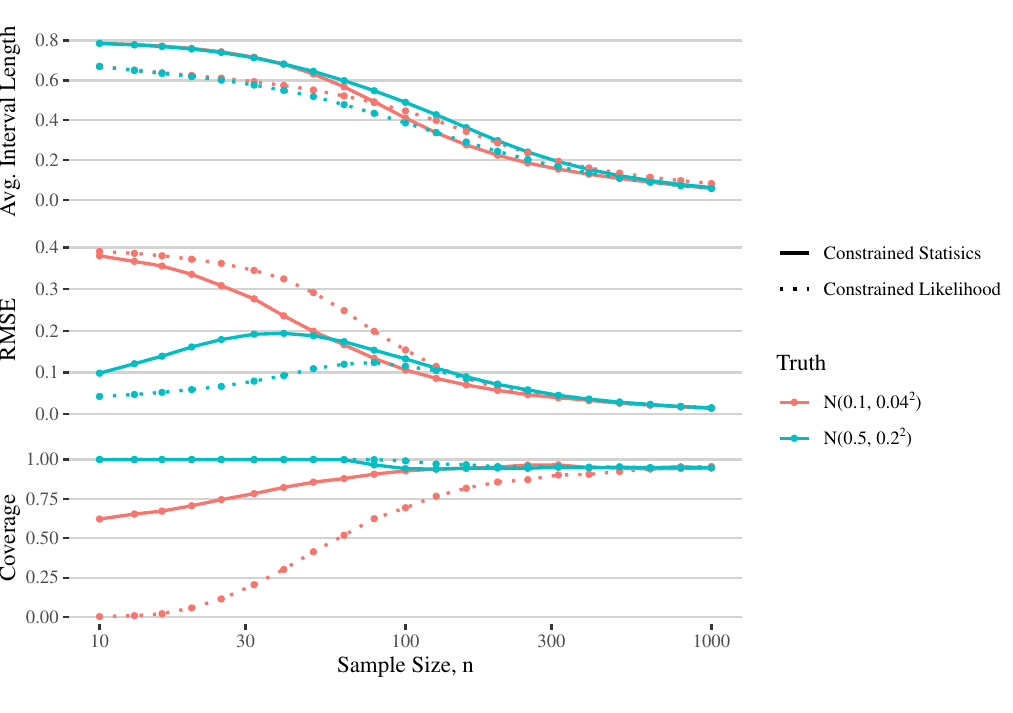}
    \caption{Average 95\% HPD interval length for $\mu$ (top), root mean square error for estimating $\mu$ (middle), and empirical coverage rate of 95\% HPD intervals for $\mu$ (bottom) for different sample sizes. Results based on 10,000 simulated datasets $Y_i \in [0,1]$ released with $\varepsilon_1 = \varepsilon_2 = 1$ and analyzed with prior $p(\mu, \sigma^2) \propto 1$. Analyses with the statistic constraints from Section \ref{sec:bounds} are solid lines and with the likelihood constraints from Section \ref{app:Ju_samp} are dotted lines. Data generating model is either $\norm(\mu = 0.1, \sigma^2 = 0.04^2)$ (red) or $\norm(\mu = 0.5, \sigma^2 = 0.2^2)$ (blue). Each Gibbs sampler is run for 20,000 iterations.}
    \label{fig:coverage_comp}
\end{figure}

\section{CONSTRAINTS IN SIMPLE LINEAR REGRESSION} \label{app:BeSh_bounds}

In this section, we derive constraints on the model proposed by \cite{bernstein2019differentially} in the simple linear regression setting. Letting $\X = [\one_n \hspace{0.08in} \x_1]$, the authors release $\X^\top \X$, $\X^\top \Y$, and $\Y^\top \Y$, which include the following elements.
\begin{align}
    \X^\top \X = \begin{bmatrix}
        \one_n^\top \one_n & \one_n^\top \x_1 \\
        \x_1^\top \one_n & \x_1^\top \x_1
    \end{bmatrix}, \qquad \X^\top \Y = \begin{bmatrix}
        \one_n^\top \Y \\ \x_1^\top \Y
    \end{bmatrix}.
\end{align}
Note that $\one_n^\top \one_n = n$ and as sums of squares, $\x_1^\top \x_1 = \sum_{i=1}^n x_i^2 \geq 0$ and $\Y^\top \Y = \sum_{i=1}^n y_i^2 \geq 0$. The requirement that $x_i, y_i \in [0,1]$ for all $i \in \{1, \ldots, n\}$ implies the following additional constraints.
\begin{enumerate}
    \item Since $x_i \leq 1$ for all $i$, it follows that $\x_1^\top \one_n = \one_n^\top \x_1 = \sum_{i=1}^n x_i \leq n$.
    \item Since $y_i \leq 1$ for all $i$, it follows that $\one_n^\top \Y = \sum_{i=1}^n y_i \leq n$.
    \item Since $x_i,y_i \geq 0$ for all $i$, it follows that $\x_1^\top \Y = \sum_{i=1}^n x_i y_i \geq 0$.
    \item Since $x_i \leq 1$ for all $i$, it follows that $\x_1^\top \Y = \sum_{i=1}^n x_i y_i \leq \sum_{i=1}^n y_i = \one_n^\top \Y$.
    \item Since $y_i \leq 1$ for all $i$, it follows that $\x_1^\top \Y = \sum_{i=1}^n x_i y_i \leq \sum_{i=1}^n x_i = \one_n^\top \x_1$.
    \item Since $0 \leq x_i \leq 1$ for all $i$, $x_i^2 \leq x_i$ and so $\x_1^\top \x_1 = \sum_{i=1}^n x_i^2 \leq \sum_{i=1}^n x_i = \one_n^\top \x_1$.
    \item Since $0 \leq y_i \leq 1$ for all $i$, $y_i^2 \leq y_i$ and so $\Y^\top \Y = \sum_{i=1}^n y_i^2 \leq \sum_{i=1}^n y_i = \one_n^\top \Y$.
\end{enumerate}

The form of the model implies constraints on the parameters $\theta_0$, $\theta_1$, and $\sigma^2$. We can rewrite the model in the following form, where for each $i \in \{1, \ldots, n\}$,
\begin{align}
    y_i \sim \norm(\theta_0 + \theta_1 x_i, \sigma^2).
\end{align}
For the regression coefficients, $\theta_0$ and $\theta_1$, we observe that by monotonicity of expectation, since $y_i \in [0,1]$, it follows that $E[y_i] = \theta_0 + \theta_1 x_i \in [0,1]$. Using this fact, for a given value of $\theta_1$, we obtain the following constraints on $\theta_0$.
\begin{enumerate}
    \item When $\theta_1 \geq 0$, since $x_i \geq 0$ for all $i$, it follows that $\theta_0 \leq \theta_0 + \theta_1 x_i \leq 1$.
    \item When $\theta_1 \geq 0$, since $x_i \leq 1$ for all $i$, it follows that $\theta_0 + \theta_1 \geq \theta_0 + \theta_1 x_i \geq 0$ and so $\theta_0 \geq -\theta_1$.
    \item When $\theta_1 < 0$, since $x_i \geq 0$ for all $i$, it follows that $\theta_0 \geq \theta_0 + \theta_1 x_i \geq 0$.
    \item When $\theta_1 < 0$, since $x_i \leq 1$ for all $i$, it follows that $\theta_0 + \theta_1 \leq \theta_0 + \theta_1 x_i \leq 1$ and so $\theta_0 \leq 1 - \theta_1$.
\end{enumerate}

We summarize these constraints as follows.

\begin{align}
    \begin{cases}
        0 \leq \theta_0 \leq 1 - \theta_1, & \mbox{if } \theta_1 < 0; \\
        -\theta_1 \leq \theta_0 \leq 1, & \mbox{if } \theta_1 \geq 0.
    \end{cases}
\end{align}

Finally, we use a similar argument to Corollary \ref{cor:sigma_sq} to obtain a constraint on $\sigma^2$. By monotonicity of expectation and since $y_i \in [0,1]$ implies $y_i^2 \leq y_i$, for all $i$,
\begin{align}
    \sigma^2 = Var[y_i] = E[y_i^2] - E[y_i]^2 \leq E[y_i] - E[y_i]^2 \leq 1/4.
\end{align}
The final inequality is obtained by observing that the quadratic $f(x) = x-x^2$ is minimized at $x = 1/2$ and has minimum $f(1/2) = 1/4$. Since the variance of $y_i$ must be nonnegative, it follows that $\sigma^2 \in [0, 1/4]$.

\section{COMPUTE DETAILS} \label{app:compute}

In this section, we describe information on the computer resources needed to reproduce the experiments. We note that Figures \ref{fig:bounding}, \ref{fig:bounding_flat}, \ref{fig:predictive}, \ref{fig:TGM}, \ref{fig:constraints_comp_1}, and \ref{fig:constraints_comp_2} can be run on a personal computer in a few seconds and do not require any additional computing resources. Figure \ref{fig:BeSh_comparison} can be run on a personal computer in under five minutes and so also does not require any additional computing resources.

Figures \ref{fig:coverage}, \ref{fig:coverage_RMSE}, and \ref{fig:coverage_2} are based on a large simulation study run on an internal shared compute cluster. The cluster was primarily used for parallelization of computations over a large number of cores; no substantial memory or storage was required. Simulations were run for two values of $\varepsilon$, two data generating distributions, and 21 sample sizes for each the unconstrained and constrained analysis. This yielded 168 input combinations, each of which was ran on a different core. The compute time required to run the Gibbs sampler on $10{,}000$ simulated datasets varied for each input combination, but was generally between 10 and 20 hours.

Figure \ref{fig:coverage_comp} is also based on the large simulation study described above. Simulations from Figures \ref{fig:coverage} and \ref{fig:coverage_RMSE} are reused for the constrained statistics values. The constrained likelihood values were again run on an internal shared compute cluster, this time with 84 input combinations. The compute time required to run the Gibbs sampler on $10{,}000$ simulated datasets increased with $n$, requiring less than 24 hours for $n \leq 50$, 5 days for $n = 500$, and 10 days for $n = 1{,}000$.  The increased compute time is due primarily to repeatedly running the modified sampler of \cite{ju2022data}, which has computational complexity $\mathcal{O}(nT)$.

\section{CODE} \label{app:code}

\texttt{R} code to produce all figures in the main document is provided below. All code, including the code to produce figures in the appendix can be found at \url{https://github.com/zekicankazan/dp_priors}. We import the following \texttt{R} packages.

\small
\begin{verbatim}
library(tidyverse)
library(purrr)
library(ggthemes)
library(rmutil)
library(truncdist)
library(truncnorm)
library(mgcv)
library(MASS)
library(patchwork)
library(Matrix)
library(matrixcalc)
library(HDInterval)
library(MCMCglmm)
\end{verbatim}

\normalsize

\normalsize
\subsection{Gibbs Sampler}
We now provide the two functions necessary to implement the Gibbs sampler described in this work. The following function draws a sample from the Truncated Gamma Mixture distribution.

\small
\begin{verbatim}
rTGM <- function(alpha, beta, lambda, tau, upper_bound = Inf){
  # Function to draw a random variable from a Truncated Gamma Mixture Distribution
  if(is.infinite(upper_bound) & tau <= 0){
    return(rgamma(1, shape = alpha, rate = beta + lambda))
  }
  else if(!is.infinite(upper_bound) & tau <= 0){
    return(rtrunc(1, "gamma", b = upper_bound, shape = alpha, rate = beta + lambda))
  }
  else if(tau >= upper_bound){
    return(rtrunc(1, "gamma", b = upper_bound, shape = alpha, rate = beta - lambda))
  }
  else{ #w_1, w_2 s.t. pi_1 = w_1/(w_1+w_2)
    log_w_1 <- -lambda*tau + pgamma(tau, alpha, beta - lambda, log.p = T) +
      lgamma(alpha) - alpha*log(beta - lambda)
    if(!is.infinite(upper_bound)){
      log_w_2 <- lambda*tau + log(diff(exp(pgamma(c(upper_bound, tau), alpha, 
                                                  beta + lambda, lower=F, log.p = T)))) +
        lgamma(alpha) - alpha*log(beta + lambda)
    }
    else{
      log_w_2 <- lambda*tau + pgamma(tau, alpha, beta + lambda, lower=F, log.p = T) +
        lgamma(alpha) - alpha*log(beta + lambda)
    }
    # Sample 1 with prob pi_1 = w_1/(w_1+w_2)
    ratio <- exp(log_w_1 - log_w_2) # ratio = w_1/w_2
    if(is.infinite(ratio)){ samp <- 1}
    else{ samp <- sample(c(1,2), 1, prob = c(ratio, 1))}
    if(samp == 1){
      return(rtrunc(1, "gamma", b = tau, shape = alpha, rate = beta - lambda))
    }
    else{ # samp == 2
      return(rtrunc(1, "gamma", a = tau, b = upper_bound, shape = alpha, rate = beta + lambda))
    }
  }
}
\end{verbatim}

\normalsize
The following function implements all combinations of the Gibbs sampler discussed in Sections \ref{sec:gibbs}, \ref{sec:Gauss_bounds}, and \ref{sec:flat_priors}.

\small
\begin{verbatim}
Gibbs <- function(nsims, n, priv_Y_bar, priv_S_sq, eps_1, eps_2, flat_prior = T, 
                  bounded = F, mu_0 = NA, kappa_0 = NA, nu_0 = NA, sigma_0_sq = NA){
  # Function to run our proposed Gibbs Sampler
  
  # Create vectors to store samples
  mu_samp <- rep(NA, nsims); sigma_sq_samp <- rep(NA, nsims); Y_bar_samp <- rep(NA, nsims); 
  S_sq_samp <- rep(NA, nsims); omega_sq_samp <- rep(NA, nsims)
  
  # Initialize parameters
  sigma_sq <- S_sq <- case_when(priv_S_sq <= 0 ~ 1e-4, priv_S_sq >= 1/4 ~ 1/4, T ~ priv_S_sq)
  Y_bar <- if_else(priv_Y_bar < 0, 0, min(priv_Y_bar, 1)); omega_sq_inv <- eps_1^2*n^2/2
  
  for(t in 1:nsims){
    # Sample mu from full conditional
    if(flat_prior & !bounded){
      mu <- rnorm(1, mean = Y_bar, sd = sqrt(sigma_sq/n))
    }
    else if(!flat_prior & !bounded){
      mu <- rnorm(1, mean = (n*Y_bar + kappa_0*mu_0)/(n+kappa_0), sd = sqrt(sigma_sq/(n+kappa_0)))
    }
    else if(sigma_sq == 1/4){
      mu <- 1/2
    }
    else if(flat_prior & bounded){
      mu <- rtruncnorm(1, a = 1/2 - sqrt(1/4 - sigma_sq), b = 1/2 + sqrt(1/4 - sigma_sq), 
                       mean = Y_bar, sd = sqrt(sigma_sq/n))
    }
    else if(!flat_prior & bounded){
      mu <- rtruncnorm(1, a = 1/2 - sqrt(1/4 - sigma_sq), b = 1/2 + sqrt(1/4 - sigma_sq), 
                       mean = (n*Y_bar + kappa_0*mu_0)/(n+kappa_0), sd = sqrt(sigma_sq/(n+kappa_0)))
    }
    mu_samp[t] <- mu
    
    # Sample sigma^2 from full conditional
    if(flat_prior & !bounded){
      sigma_sq <- 1/rtrunc(1, "gamma", a = n/(n-1)*2*eps_2, shape = (n-2)/2, 
                           rate = ((n-1)*S_sq + n*(Y_bar - mu)^2)/2)
    }
    else if(!flat_prior & !bounded){
      sigma_sq <- 1/rtrunc(1, "gamma", a = n/(n-1)*2*eps_2, shape = (n+nu_0)/2, 
                           rate = (nu_0*sigma_0_sq + (n-1)*S_sq + n*(Y_bar - mu)^2)/2)
    }
    else if(flat_prior & bounded){
      sigma_sq <- 1/rtrunc(1, "gamma", a = max(1/(mu*(1-mu)), 2*eps_2*n/(n-1)), 
                           shape = (n-2)/2, rate = ((n-1)*S_sq + n*(Y_bar - mu)^2)/2)
    }
    else if(!flat_prior & bounded){
      sigma_sq <- 1/rtrunc(1, "gamma", a = max(1/(mu*(1-mu)), 2*eps_2*n/(n-1)), shape = (n+nu_0)/2, 
                           rate = (nu_0*sigma_0_sq + (n-1)*S_sq + n*(Y_bar - mu)^2)/2)
    }
    sigma_sq_samp[t] <- sigma_sq
    
    # Sample bar{Y} from full conditional
    if(!bounded){ 
      Y_bar <- rnorm(1, sd = sqrt(1/(omega_sq_inv + n/sigma_sq)),
                     mean = (priv_Y_bar*omega_sq_inv + n*mu/sigma_sq)/(omega_sq_inv + n/sigma_sq))
    }
    else{
      Y_bar <- rtruncnorm(1, a = 1/2 - sqrt(1/4 - (n-1)/n*S_sq), b = 1/2 + sqrt(1/4 - (n-1)/n*S_sq), 
                          mean = (priv_Y_bar*omega_sq_inv + n*mu/sigma_sq)/(omega_sq_inv + n/sigma_sq),
                          sd = sqrt(1/(omega_sq_inv + n/sigma_sq)))
    }
    Y_bar_samp[t] <- Y_bar
    
    # Sample omega^2 from full conditional
    omega_sq_inv <- rig(1, mean = eps_1*n/abs(priv_Y_bar - Y_bar), scale = 1/(eps_1^2*n^2))
    omega_sq_samp[t] <- omega_sq_inv
    
    # Sample S^2 from full conditional
    if(!bounded){
      S_sq <- rTGM(alpha = (n-1)/2, beta = (n-1)/(2*sigma_sq), lambda = eps_2*n,  tau = priv_S_sq)
    }
    else{
      S_sq <- rTGM(alpha = (n-1)/2, beta = (n-1)/(2*sigma_sq), lambda = eps_2*n, 
                   tau = priv_S_sq, upper_bound = n/(n-1)*Y_bar*(1-Y_bar))
    }
    S_sq_samp[t] <- S_sq
  }
  
  df <- data.frame(t = 1:nsims, mu = mu_samp, sigma_sq = sigma_sq_samp, 
                   Y_bar = Y_bar_samp, omega_sq = omega_sq_samp, S_sq = S_sq_samp,
                   bounded = bounded, flat_prior = flat_prior)
  return(df)
}
\end{verbatim}

\normalsize
\subsection{Figures \ref{fig:bounding} and \ref{fig:bounding_flat}}

The following code uses the above functions to produce Figures \ref{fig:bounding} and \ref{fig:bounding_flat}.

\small
\begin{verbatim}
set.seed(6)
a <- 0; b <- 100
true_Y_bar <- (32.08 - a)/(b-a); true_S_sq <- (16.98/(b-a))^2; nsims <- 5e3
eps_1 <- 0.25; eps_2 <- 0.25; n <- 43
priv_Y_bar <- rlaplace(1, true_Y_bar, 1/(eps_1*n)); priv_S_sq <- rlaplace(1, true_S_sq, 1/(eps_2*n)) 

mu_0 <- (12.5 - a)/(b-a); sigma_0_sq <- (3.8/(b-a))^2; kappa_0 <- 1; nu_0 <- 1

# Run each of the four Gibbs Samplers
df_bd <- Gibbs(nsims, n, priv_Y_bar, priv_S_sq, eps_1, eps_2, bounded = T,flat_prior = F, 
               mu_0 = mu_0, kappa_0 = kappa_0, nu_0 = nu_0, sigma_0_sq = sigma_0_sq)
df_ub <- Gibbs(nsims, n, priv_Y_bar, priv_S_sq, eps_1, eps_2, bounded = F, flat_prior = F, 
               mu_0 = mu_0, kappa_0 = kappa_0, nu_0 = nu_0, sigma_0_sq = sigma_0_sq)
df_bd_flat <- Gibbs(nsims, n, priv_Y_bar, priv_S_sq, eps_1, eps_2, bounded = T, flat_prior = T)
df_ub_flat <- Gibbs(nsims, n, priv_Y_bar, priv_S_sq, eps_1, eps_2, bounded = F, flat_prior = T)

# Create Shaded Region
df_fill <- mutate(data.frame(mu = b*seq(0,1,0.005), sigma_sq = 0), ymin = 0, ymax = mu*(b-mu))

# Compute the Four Shapes for Figure 1
add_points <- data.frame(name = "Posterior Mode", bounded = c(T, F),
                         mu = c(posterior.mode(df_bd$mu, adjust = 1), 
                                posterior.mode(df_ub$mu, adjust = 1)),
                         sigma_sq = c(posterior.mode(df_bd$sigma_sq, adjust = 1), 
                                      posterior.mode(df_ub$sigma_sq, adjust = 1))) %>%
  rbind(data.frame(name = "Prior Values", bounded = c(T,F), mu = mu_0, sigma_sq = sigma_0_sq)) %>%
  rbind(data.frame(name = "Confidential Values", bounded = c(T,F), 
                   mu = true_Y_bar, sigma_sq = true_S_sq)) %>%
  rbind(data.frame(name = "Released Values", bounded = c(T,F), 
                   mu = priv_Y_bar, sigma_sq = priv_S_sq)) %>%
  mutate(bounded = factor(bounded, c(F, T), c("Unconstrained", "Constrained")),
         name = factor(name, levels = c("Released Values", "Confidential Values",
                                        "Posterior Mode", "Prior Values")),
         mu = mu*b, sigma_sq = sigma_sq*b^2)

# Compute the Three Shapes for Figure 3
add_points_flat <- data.frame(name = "Posterior Mode", bounded = c(T, F),
                              mu = c(posterior.mode(df_bd_flat$mu, adjust = 1), 
                                     posterior.mode(df_ub_flat$mu, adjust = 1)),
                              sigma_sq = c(posterior.mode(df_bd_flat$sigma_sq, adjust = 1), 
                                           posterior.mode(df_ub_flat$sigma_sq, adjust = 1))) %>%
  rbind(data.frame(name = "Confidential Values", bounded = c(T,F), 
                   mu = true_Y_bar, sigma_sq = true_S_sq)) %>%
  rbind(data.frame(name = "Released Values", bounded = c(T,F), 
                   mu = priv_Y_bar, sigma_sq = priv_S_sq)) %>%
  mutate(bounded = factor(bounded, c(F, T), c("Unconstrained", "Constrained")),
         name = factor(name, levels = c("Released Values", "Confidential Values",
                                        "Posterior Mode", "Prior Values")),
         mu = mu*(b-a)+a, sigma_sq = sigma_sq*(b-a)^2)

# Create Figure 1
df_bd %>%
  rbind(df_ub) %>%
  mutate(bounded = factor(bounded, c(F, T), c("Unconstrained", "Constrained")),
         mu = mu*b, sigma_sq = sigma_sq*b^2) %>%
  ggplot(aes(x = mu, y = sigma_sq)) +
  geom_point(alpha = 0.3, size = 1) + 
  geom_ribbon(data = df_fill, aes(ymin = ymin, ymax = ymax), fill = "gray", alpha = 0.5) +
  geom_point(aes(x = mu, y = sigma_sq, shape = name, color = name), data = add_points, size = 3) +
  facet_grid(rows = vars(bounded)) +
  scale_shape_manual(values=c(16,17,18,15)) +
  labs(x = expression("Average Blood Lead ("*mu*"g/dL)"), shape = "", color = "",
       y = expression("Blood Lead Variance ("*mu*"g"^2*"/dL"^2*")")) +
  theme_tufte() +
  theme(plot.margin=grid::unit(c(0,0,0,0), "mm"), legend.position="bottom",
        legend.margin=margin(0,0,0,0), legend.box.margin=margin(-5,0,0,0),
        legend.spacing.y = unit(-0.5, "mm")) +
  guides(colour=guide_legend(ncol=2,nrow=2,byrow=TRUE), shape = guide_legend(ncol=2,nrow=2,byrow=TRUE))

# Create Figure 3
df_bd_flat %>%
  rbind(df_ub_flat) %>%
  mutate(bounded = factor(bounded, c(F, T), c("Unconstrained", "Constrained")),
         mu = mu*b, sigma_sq = sigma_sq*b^2) %>%
  ggplot(aes(x = mu, y = sigma_sq)) +
  geom_point(alpha = 0.3, size = 1) + 
  geom_ribbon(data = df_fill, aes(ymin = ymin, ymax = ymax), fill = "gray", alpha = 0.5) +
  geom_point(aes(x = mu, y = sigma_sq, shape = name, color = name), data = add_points_flat, size = 3) +
  facet_grid(rows = vars(bounded)) +
  scale_y_continuous(limits = c(NA, 10000)) +
  scale_shape_manual(values=c(16,17,18,15)) +
  scale_color_manual(values = c("#F8766D", "#7CAE00", "#00BFC4")) +
  labs(x = expression("Average Blood Lead ("*mu*"g/dL)"), shape = "", color = "",
       y = expression("Blood Lead Variance ("*mu*"g"^2*"/dL"^2*")")) +
  theme_tufte() +
  theme(plot.margin=grid::unit(c(0,0,0,0), "mm"), legend.position="bottom",
        legend.margin=margin(0,0,0,0), legend.box.margin=margin(-5,0,0,0),
        legend.spacing.y = unit(-0.5, "mm")) +
  guides(colour=guide_legend(ncol=2,nrow=2,byrow=TRUE), shape = guide_legend(ncol=2,nrow=2,byrow=TRUE))
\end{verbatim}

\normalsize
\subsection{Figure \ref{fig:coverage}}

The following code uses the above functions to produce Figure \ref{fig:coverage}. The first section of code, which we run on a shared compute cluster, creates a dataframe with the relevant quantities for desired combinations of $n$, $\varepsilon$, $(\mu, \sigma^2)$, and whether the constraint is enforced.

{\small
\begin{verbatim}
n_vec <- round(10^seq(1,3, 0.1)); eps_vec <- c(0.2, 2); 
bdd_vec <- c(FALSE, TRUE); mu_vec <- c(0.1, 0.5)
grid <- expand.grid(n = n_vec, eps = eps_vec, bdd = bdd_vec, mu = mu_vec)
reps <- 10000; nsims <- 20000

len <- nrow(grid)
plot_df <- mutate(grid, coverage = NA, CI_len = NA, RMSE = NA)
for(num in 1:len){
  n <- grid$n[num]; eps <- grid$eps[num]; bounded <- grid$bdd[num]; mu <- grid$mu[num]
  sigma_sq <- (mu*0.4)^2; eps_1 <- eps_2 <- eps/2
  
  df <- data.frame(posterior_mode = rep(NA, reps), lower = NA, upper = NA)
  set.seed(2024)
  for(i in 1:reps){
    Y <- rtruncnorm(n, a = 0, b = 1, mean = mu, sd = sqrt(sigma_sq))
    priv_Y_bar <- rlaplace(1, mean(Y), 1/(eps_1*n))
    priv_S_sq <- rlaplace(1, var(Y), 1/(eps_2*n))
    
    out <- Gibbs(nsims, n, priv_Y_bar, priv_S_sq, eps_1, eps_2, flat_prior = T, 
                bounded)
    
    df$posterior_mode[i] <- posterior.mode(out$mu, adjust = 1)
    df[i,c("lower", "upper")] <- hdi(out$mu, credMass = 0.95)
  }
  plot_df$CI_len[num] <- mean(df$upper - df$lower)
  plot_df$coverage[num] <- mean(mu <= df$upper & mu >= df$lower)
  plot_df$RMSE[num] <- sqrt(mean((df$posterior_mode - mu)^2))
}
\end{verbatim}
}

The next section of code uses the computed quantities to create the figure.

{\small
\begin{verbatim}
plot_eps <- 0.2

# Create Top Panel
p1 <- plot_df %>%
  filter(n >= 30, eps == plot_eps) %>%
  mutate(bdd = factor(bdd, c(T, F))) %>%
  ggplot(aes(x = n, y = CI_len, color = factor(mu), linetype = bdd)) +
  geom_line(linewidth = 0.75) + geom_point(size = 1) +
  labs(color = "Truth", linetype = "Constrained?", x = "Sample Size, n", y = "Avg. Interval Length") +
  scale_x_log10() +
  scale_color_discrete(labels = c(expression("N(0.1, 0.04"^2*")"), expression("N(0.5, 0.2"^2*")  "))) +
  theme_tufte()  +
  theme(axis.title.x=element_blank(), axis.text.x=element_blank(), axis.ticks.x=element_blank())

# Create Bottom Panel
p2 <- plot_df %>%
  filter(n >= 30, eps == plot_eps) %>%
  mutate(bdd = factor(bdd, c(T, F))) %>%
  ggplot(aes(x = n, y = coverage, color = factor(mu), linetype = bdd)) +
  geom_line(linewidth = 0.75) + geom_point(size = 1) +
  labs(color = "Truth", linetype = "Constrained?", x = "Sample Size, n", y = "Coverage") +
  scale_x_log10() +
  scale_color_discrete(labels = c(expression("N(0.1, 0.04"^2*")"), expression("N(0.5, 0.2"^2*")  "))) +
  theme_tufte()

p1/p2 + plot_layout(guides = "collect") & 
    theme(plot.margin=grid::unit(c(1,0.1,4,0), "mm"), legend.position="bottom",
          legend.margin=margin(0,0,0,0), legend.box.margin=margin(-5,0,-25,0),
          legend.box="vertical", legend.spacing.y = unit(0, "mm"),
          panel.grid.major.y = element_line(color = "lightgray", linewidth = 0.5))
\end{verbatim}
}

\subsection{Figure \ref{fig:predictive}}

The following code uses the above functions to produce Figure \ref{fig:predictive}.

{\small
\begin{verbatim}
set.seed(6)
a <- 0; b <- 100; 
true_Y_bar <- (32.08 - a)/(b-a); true_S_sq <- (16.98/(b-a))^2; nsims <- 1e5
eps_1 <- 0.25; eps_2 <- 0.25; n <- 43; 
priv_Y_bar <- rlaplace(1, true_Y_bar, 1/(eps_1*n)); 
priv_S_sq <- rlaplace(1, true_S_sq, 1/(eps_2*n)) 

df_bd_flat <- Gibbs(nsims, n, priv_Y_bar, priv_S_sq, eps_1, eps_2, bounded = T, flat_prior = T)
df_ub_flat <- Gibbs(nsims, n, priv_Y_bar, priv_S_sq, eps_1, eps_2, bounded = F, flat_prior = T)

df_fill_flat <- mutate(data.frame(preds = seq(a,b)), ymin = 0, ymax = Inf)

preds_ub <- rnorm(nsims, df_ub_flat$mu, sqrt(df_ub_flat$sigma_sq))
preds_bd <- rtruncnorm(nsims, a = 0, b = 1, df_bd_flat$mu, sqrt(df_bd_flat$sigma_sq))

add_lines <- data.frame(name = "Posterior Mode", bounded = c(T, F), 
                        mu = c(posterior.mode(df_bd_flat$mu, adjust = 1), 
                               posterior.mode(df_ub_flat$mu, adjust = 1))) %>%
  rbind(data.frame(name = "Confidential Value", bounded = c(T,F), mu = true_Y_bar)) %>%
  rbind(data.frame(name = "Released Value", bounded = c(T,F), mu = priv_Y_bar)) %>%
  mutate(bounded = factor(bounded, c(F, T), c("Unconstrained", "Constrained")),
         name = factor(name, levels = c("Released Value", "Confidential Value", "Posterior Mode")),
         mu = mu*b)

data.frame(preds = c(preds_ub, preds_bd),
           bounded = c(rep("Unconstrained", nsims), rep("Constrained", nsims))) %>%
  mutate(bounded = factor(bounded, c("Unconstrained", "Constrained")), preds = preds*(b-a) + a) %>%
  ggplot(aes(x = preds)) + 
  geom_histogram(binwidth = 0.1*b) +
  geom_ribbon(data = df_fill_flat, aes(ymin = ymin, ymax = ymax), fill = "gray", alpha = 0.5) +
  facet_grid(bounded~ .) +
  scale_x_continuous(limits = c(-1*b, 2*b)) +
  labs(x = expression("New Observation's Blood Lead ("*mu*"g/dL)"),
       color = "", linetype = "", y = "Count") +
  theme_tufte() +
  theme(plot.margin=grid::unit(c(0,0,0,0), "mm"), legend.position="bottom",
        legend.margin=margin(0,0,0,0), legend.box.margin=margin(-5,0,0,0)) +
  guides(colour=guide_legend(ncol=2,nrow=2,byrow=TRUE),
         linetype=guide_legend(ncol=2,nrow=2,byrow=TRUE))
\end{verbatim}
}

\subsection{Figure \ref{fig:BeSh_comparison}}

Finally, we provide code to produce Figure \ref{fig:BeSh_comparison}. The following function implements the method proposed by \cite{bernstein2019differentially} with and without constraints enforced.

{\small
\begin{verbatim}
Gibbs_BeSh <- function(Z, n, d, eps, Delta, eta, xi, mu_0, Lambda_0, a_0, b_0, nsims, bounded = F){
  # Implementation of Bernstein & Sheldon (2018)'s Gibbs-SS-Noisy Method
  
  #### Z = px1, Delta = px1, eta = dxd, xi = d*(d+1)/2 x d*(d+1)/2, mu = dx1, Lambda_0 = dxd
  #### All others constants
  p <- length(Z) # Number of releases
  
  # Create vectors to store samples
  theta_samp <- matrix(NA, nsims, d); sigma_sq_samp <- rep(NA, nsims);
  omega_sq_inv_samp <- matrix(NA, nsims, p); S_samp <- matrix(NA, nsims, p)
  
  # Initialize Parameters
  theta <- mu_0; sigma_sq <- b_0; omega_sq_inv <- (eps/p)^2/Delta^2/2
  
  for(t in 1:nsims){
    # Sample s from full Conditional
    ### Compute mu_t and Sigma_t, hardcoded for d = 2
    if(!bounded){
      mu_t <- c(eta[lower.tri(eta,diag=T)], #Expectation for XTX
                theta %*% eta, #Expectation for XTY
                sigma_sq + theta %*% eta %*% theta) #Expectation for YTY
      active <- c(1:2, 4)
    }
    else if(bounded){
      mu_t <- c(eta[,2], #Expectation for XTX
                theta %*% eta, #Expectation for XTY
                sigma_sq + theta %*% eta %*% theta) #Expectation for YTY
      active <- 3:4
    }
    Sigma_11 <- xi[active, active]
    Sigma_12 <- matrix(matrix(xi,8,2) %*% theta,4,2)[active,]
    Sigma_13 <- (xi %*% matrix(theta %*% t(theta)))[active,]
    Sigma_22 <- sigma_sq*eta + matrix(xi %*% matrix(theta %*% t(theta)),ncol = 2)
    Sigma_23 <- 2*sigma_sq*eta %*% theta  + 
      matrix(xi,2,8,T) %*% matrix(matrix(theta %*% t(theta)) %*% theta)
    Sigma_33 <- 2*sigma_sq^2 + 4*sigma_sq*theta %*% eta %*% theta +
      matrix(xi,1) %*% matrix(matrix(theta %*% t(theta)) %*% matrix(theta %*% t(theta),1))
    Sigma_t <- nearPD(rbind(cbind(Sigma_11, Sigma_12, Sigma_13),
                            cbind(t(Sigma_12), Sigma_22, Sigma_23),
                            c(Sigma_13, t(Sigma_23), Sigma_33)))$mat
    
    ### Sample S
    Sigma_3 <- solve(solve(Sigma_t)/n + diag(omega_sq_inv))
    mu_3 <- Sigma_3 %*% (solve(Sigma_t) %*% mu_t + diag(omega_sq_inv) %*% Z)
    if(!bounded){
      S <- mvrnorm(1, mu = mu_3, Sigma = Sigma_3)
    }
    else if(bounded){
      flag <- F
      while(!flag){
        S <- mvrnorm(1, mu = mu_3, Sigma = Sigma_3)
        B <- matrix(c(n, S[c(1,3,1,2,4,3:5)]),3)
        flag <- (S[1] >= 0 & S[1] <= n) & (S[2] >= 0 & S[2] <= S[1]) &
          (S[3] >= 0 & S[3] <= n) & (S[4] >= 0 & S[4] <= S[1] & S[4] <= S[3]) &
          (S[5] >= 0 & S[5] <= S[3]) & is.positive.definite(B)
      }}
    S_samp[t,] <- as.vector(S)
    
    # Sample theta, sigma_sq from full Conditional
    ### Obtain current sufficient statistics
    XTX <- matrix(NA,d,d); 
    if(!bounded){
      XTX[lower.tri(t(XTX), diag = T)] <- S[1:(d*(d+1)/2)]
      XTY <- matrix(S[(d*(d+1)/2 + 1):(d*(d+1)/2 + d)],d,1)
    }
    else if(bounded){
      XTX[lower.tri(t(XTX), diag = T)] <- c(n, S[1:(d*(d+1)/2 - 1)])
      XTY <- matrix(S[(d*(d+1)/2):(d*(d+1)/2 + d - 1)],d,1)
    }
    XTX[upper.tri(XTX)] <- (t(XTX))[upper.tri(XTX)]
    YTY <- S[length(S)]
    
    if(!bounded){
      B <- as.matrix(nearPD(rbind(cbind(XTX, XTY), cbind(t(XTY), YTY)))$mat)
      XTX <- B[1:d, 1:d]; XTY <- matrix(B[1:d, d+1]); YTY <- B[d+1, d+1]
    }
    
    ### Compute hyperparameters
    mu_n <- solve(XTX + Lambda_0) %*% (XTY + Lambda_0 %*% mu_0)
    Lambda_n <- XTX + Lambda_0
    a_n <- a_0 + n/2
    b_n <- b_0 + (YTY + t(mu_0) %*% Lambda_0 %*% mu_0 - as.numeric(t(mu_n) %*% Lambda_n %*% mu_n))/2
    
    ### Sample
    if(!bounded){
      sigma_sq <- 1/rgamma(1, shape = a_n,  rate = b_n)
      theta <- mvrnorm(1, mu = mu_n, Sigma = solve(nearPD(Lambda_n)$mat)*sigma_sq)
    }
    else if(bounded){
      flag = F
      while(!flag){
        sigma_sq <- 1/rtrunc(1, "gamma", b = 1/4, shape = a_n,  rate = b_n)
        theta <- mvrnorm(1, mu = mu_n, Sigma = solve(nearPD(Lambda_n)$mat)*sigma_sq)
        flag <- (theta[1] >= 0 & theta[1] <= 1) | (theta[1] > 1 & theta[2] <= 1 - theta[1]) |
          (theta[1] < 0 & theta[2] >= -theta[1])
      }}
    
    ### Save samples
    sigma_sq_samp[t] <- sigma_sq
    theta_samp[t,] <- theta
    
    # Sample omega_sq_inv from full conditional
    omega_sq_inv <- rig(p, mean = as.vector(eps/p/Delta/abs(Z - S)), scale = Delta^2/(eps/p)^2)
    omega_sq_inv_samp[t,] <- omega_sq_inv
  }
  return(list(t = 1:t, theta = theta_samp, sigma_sq = sigma_sq_samp, 
              omega_sq_inv = omega_sq_inv_samp, S = S_samp))
}
\end{verbatim}
}

The following code uses the above function to produce the figure.

{\small
\begin{verbatim}
# Import Data
x20 <- read.csv("https://people.sc.fsu.edu/~jburkardt/datasets/regression/x20.txt", 
                sep = "", header = F, comment.char = "#",
                col.names = c("ID", "Intercept", "pop_size", "births", "wine", "liquor", "death")) %>%
  tail(-9) %>%
  remove_rownames() %>%
  column_to_rownames(var = "ID") %>%
  mutate(across(.fns = as.numeric))

# Compute X and Y
n <- nrow(x20)
Y <- (x20$death - min(x20$death))/(max(x20$death) - min(x20$death))
drink <- x20$wine + x20$liquor; x1 <- (drink - min(drink))/(max(drink) - min(drink))
X <- matrix(c(x20$Intercept, x1), ncol = 2)

# Compute Noisy Statistics
set.seed(10)
eps <- 0.1*11; n <- nrow(X); d <- ncol(X); p <- d*(d+1)/2 + d + 1
Delta <- rep(1, p); nsims <- 5000
mu_0 <- c(1,0); Lambda_0 <- diag(c(.25, .25)); a_0 <- 20; b_0 <- 0.5

XTX_true <- t(X) %*% X; XTY_true <- t(X) %*% Y; YTY_true <- t(Y) %*% Y
S_true <- c(XTX_true[lower.tri(XTX_true,diag=T)], XTY_true, YTY_true)
Z <- rlaplace(p, S_true, Delta/(eps/p))

X_4th_priv <- rlaplace(n = 5, m = map_dbl(.x = 0:4, .f = function(i){sum(X[,2]^i)}), s = 1/(eps/11))

# Compute Quantities from Bernstein & Sheldon (2018)
eta <- matrix(X_4th_priv[c(1,2,2,3)],2,2)/n
eta_ijkl <- matrix(X_4th_priv[c(1,2,2,3,2,3,3,4,2,3,3,4,3,4,4,5)], 4, 4)/n
xi <- eta_ijkl - as.vector(eta) %*% t(as.vector(eta))

# Run Both Algorithms
l = Gibbs_BeSh(Z, n, d, eps*6/11, Delta, eta, xi, mu_0, Lambda_0, a_0, b_0, nsims)
l_bd = Gibbs_BeSh(Z[2:length(Z)], n, d, eps*6/11, Delta[2:length(Z)], eta, xi, 
                  mu_0, Lambda_0, a_0, b_0, nsims, bounded = T)

# Create Shaded Region for Left Panels
df_fill_dat <- expand.grid(YT1 = seq(0, 42, 1), YTY = seq(0, 42, 1)) %>%
  mutate(ymax = YT1*(YT1 >= 0 & YT1 <= n), ymin = 0)

# Find Infeasible Points in Left Panels
df_infeasible_dat <- as.data.frame(l$S[,2:6]) %>%
  mutate(bdd = F) %>% 
  rename(X1T1 = V1, X1TX1 = V2, YT1 = V3, X1TY = V4, YTY = V5) %>%
  mutate(bdd = factor(bdd, c(F, T), c("Unconstrained", "Constrained"))) %>%
  filter(!(0 < YTY & YTY < YT1 & YT1 < n))

# Create Left Panels
p1 <- rbind(mutate(as.data.frame(l$S[,2:6]), bdd = F),
            mutate(as.data.frame(l_bd$S), bdd = T)) %>%
  rename(X1T1 = V1, X1TX1 = V2, YT1 = V3, X1TY = V4, YTY = V5) %>%
  mutate(bdd = factor(bdd, c(F, T), c("Unconstrained", "Constrained"))) %>%
  ggplot(aes(x = YT1, y = YTY)) + 
  geom_point(size = 1, alpha = 0.3) +
  geom_point(data = df_infeasible_dat, aes(x = YT1, y = YTY), color = "#F8766D",
             size = 1, alpha = 0.5) +
  scale_y_continuous(breaks = seq(0, 90, 30)) +
  facet_grid(bdd~.) +
  geom_ribbon(data = df_fill_dat, aes(ymin = ymin, ymax = ymax), fill = "gray", alpha = 0.5) +
  labs(x = expression("Y"^"T"*"1"), y = expression("Y"^"T"*"Y")) +
  theme_tufte() +
  theme(strip.background = element_blank(), strip.text.y = element_blank())

# Create Shaded Region for Right Panels
df_fill_pars <- expand.grid(Slope = seq(-4, 4, 0.04), Intercept = seq(-5,5, 0.05)) %>%
  mutate(ymax = 1*(Slope >= 0) + (1-Slope)*(Slope < 0), ymin = (-Slope)*(Slope >= 0))

# Find Infeasible Points in Right Panels
df_infeasible_pars <- as.data.frame(l$theta) %>%
  mutate(bdd = F) %>% 
  rename(Intercept = V1, Slope = V2) %>%
  mutate(bdd = factor(bdd, c(F, T), c("Unconstrained", "Constrained"))) %>%
  filter((Slope >= 0) & !(Intercept > -Slope & Intercept < 1) |
           (Slope < 0) & !(Intercept > 0 & Intercept < 1-Slope))

# Create Right Panels
p2 <- rbind(mutate(as.data.frame(l$theta), bdd = F), mutate(as.data.frame(l_bd$theta), bdd = T)) %>%
  rename(Intercept = V1, Slope = V2) %>%
  mutate(bdd = factor(bdd, c(F, T), c("Unconstrained", "Constrained"))) %>%
  ggplot(aes(x = Slope, y = Intercept)) +
  geom_point(size = 1, alpha = 0.3) +
  geom_point(data = df_infeasible_pars, aes(x = Slope, y = Intercept), color = "#F8766D",
             size = 1, alpha = 0.5) +
  facet_grid(bdd~.) +
  xlim(-4,4) +
  labs(x = expression(theta[1]), y = expression(theta[0])) +
  geom_ribbon(data = df_fill_pars, aes(ymin = ymin, ymax = ymax), fill = "gray", alpha = 0.5) + 
  theme_tufte()

# Create Figure 5
p1 + p2 & theme(plot.margin=grid::unit(c(0.5,0,0,0), "mm"))
\end{verbatim}
}

%% file: main.bbl
\begin{thebibliography}{40}
\providecommand{\natexlab}[1]{#1}
\providecommand{\url}[1]{\texttt{#1}}
\expandafter\ifx\csname urlstyle\endcsname\relax
  \providecommand{\doi}[1]{doi: #1}\else
  \providecommand{\doi}{doi: \begingroup \urlstyle{rm}\Url}\fi

\bibitem[Bernstein \& Sheldon(2018)Bernstein and Sheldon]{bernstein2018differentially}
Garrett Bernstein and Daniel~R Sheldon.
\newblock Differentially private {B}ayesian inference for exponential families.
\newblock \emph{Advances in Neural Information Processing Systems}, 31, 2018.

\bibitem[Bernstein \& Sheldon(2019)Bernstein and Sheldon]{bernstein2019differentially}
Garrett Bernstein and Daniel~R Sheldon.
\newblock Differentially private {B}ayesian linear regression.
\newblock \emph{Advances in Neural Information Processing Systems}, 32, 2019.

\bibitem[Brownlee(1965)]{data_x20}
Kenneth~Alexander Brownlee.
\newblock Statistical theory and methodology in science and engineering.
\newblock \url{https://people.sc.fsu.edu/~jburkardt/datasets/regression/x20.txt}, 1965.
\newblock Accessed on 4 May 2024.

\bibitem[Couch et~al.(2019)Couch, Kazan, Shi, Bray, and Groce]{couch2019differentially}
Simon Couch, Zeki Kazan, Kaiyan Shi, Andrew Bray, and Adam Groce.
\newblock Differentially private nonparametric hypothesis testing.
\newblock In \emph{Proceedings of the 2019 ACM SIGSAC Conference on Computer and Communications Security}, pp.\  737--751, 2019.

\bibitem[Diez et~al.(2012)Diez, Barr, and Cetinkaya-Rundel]{diez2012openintro}
David~M Diez, Christopher~D Barr, and Mine Cetinkaya-Rundel.
\newblock \emph{OpenIntro statistics}, volume~4.
\newblock OpenIntro Boston, MA, USA:, 2012.

\bibitem[Dimitrakakis et~al.(2017)Dimitrakakis, Nelson, Zhang, Mitrokotsa, and Rubinstein]{dimitrakakis2017differential}
Christos Dimitrakakis, Blaine Nelson, Zuhe Zhang, Aikaterini Mitrokotsa, and Benjamin~IP Rubinstein.
\newblock Differential privacy for {B}ayesian inference through posterior sampling.
\newblock \emph{Journal of Machine Learning Research}, 18\penalty0 (11):\penalty0 1--39, 2017.

\bibitem[D'Orazio et~al.(2015)D'Orazio, Honaker, and King]{d2015differential}
Vito D'Orazio, James Honaker, and Gary King.
\newblock Differential privacy for social science inference.
\newblock \emph{Sloan Foundation Economics Research Paper}, \penalty0 (2676160), 2015.

\bibitem[Du et~al.(2020)Du, Foot, Moniot, Bray, and Groce]{du2020differentially}
Wenxin Du, Canyon Foot, Monica Moniot, Andrew Bray, and Adam Groce.
\newblock Differentially private confidence intervals.
\newblock \emph{arXiv preprint arXiv:2001.02285}, 2020.

\bibitem[Dwork et~al.(2006)Dwork, McSherry, Nissim, and Smith]{dwork2006calibrating}
Cynthia Dwork, Frank McSherry, Kobbi Nissim, and Adam Smith.
\newblock Calibrating noise to sensitivity in private data analysis.
\newblock In \emph{Theory of Cryptography: Third Theory of Cryptography Conference, TCC 2006, New York, NY, USA, March 4-7, 2006. Proceedings 3}, pp.\  265--284. Springer, 2006.

\bibitem[Ferrando et~al.(2022)Ferrando, Wang, and Sheldon]{ferrando2022parametric}
Cecilia Ferrando, Shufan Wang, and Daniel Sheldon.
\newblock Parametric bootstrap for differentially private confidence intervals.
\newblock In \emph{International Conference on Artificial Intelligence and Statistics}, pp.\  1598--1618. PMLR, 2022.

\bibitem[Gelman(2006)]{gelman2006prior}
Andrew Gelman.
\newblock Prior distributions for variance parameters in hierarchical models (comment on article by {B}rowne and {D}raper).
\newblock \emph{Bayesian Analysis}, 1\penalty0 (3):\penalty0 515--534, 2006.

\bibitem[Geumlek et~al.(2017)Geumlek, Song, and Chaudhuri]{geumlek2017renyi}
Joseph Geumlek, Shuang Song, and Kamalika Chaudhuri.
\newblock R\'enyi differential privacy mechanisms for posterior sampling.
\newblock \emph{Advances in Neural Information Processing Systems}, 30, 2017.

\bibitem[Gong(2022)]{gong2022exact}
Ruobin Gong.
\newblock Exact inference with approximate computation for differentially private data via perturbations.
\newblock \emph{Journal of Privacy and Confidentiality}, 12\penalty0 (2), 2022.

\bibitem[Hoff(2009)]{hoff2009first}
Peter~D Hoff.
\newblock \emph{A First Course in Bayesian Statistical Methods}, volume 580.
\newblock Springer, 2009.

\bibitem[Jewson et~al.(2024)Jewson, Ghalebikesabi, and Holmes]{jewson2024differentially}
Jack~E Jewson, Sahra Ghalebikesabi, and Chris~C Holmes.
\newblock Differentially private statistical inference through beta-divergence one posterior sampling.
\newblock \emph{Advances in Neural Information Processing Systems}, 36, 2024.

\bibitem[Ju et~al.(2022)Ju, Awan, Gong, and Rao]{ju2022data}
Nianqiao Ju, Jordan Awan, Ruobin Gong, and Vinayak Rao.
\newblock Data augmentation mcmc for {B}ayesian inference from privatized data.
\newblock \emph{Advances in neural information processing systems}, 35:\penalty0 12732--12743, 2022.

\bibitem[Karwa \& Vadhan(2017)Karwa and Vadhan]{karwa2017finite}
Vishesh Karwa and Salil Vadhan.
\newblock Finite sample differentially private confidence intervals.
\newblock \emph{arXiv preprint arXiv:1711.03908}, 2017.

\bibitem[Kasiviswanathan \& Smith(2014)Kasiviswanathan and Smith]{kasiviswanathan2014semantics}
Shiva~P Kasiviswanathan and Adam Smith.
\newblock On the'semantics' of differential privacy: A {B}ayesian formulation.
\newblock \emph{Journal of Privacy and Confidentiality}, 6\penalty0 (1), 2014.

\bibitem[Kazan \& Reiter(2023)Kazan and Reiter]{kazan2023prior}
Zeki Kazan and Jerome~P Reiter.
\newblock Prior-itizing privacy: A {B}ayesian approach to setting the privacy budget in differential privacy.
\newblock \emph{arXiv preprint arXiv:2306.13214}, 2023.

\bibitem[Kazan \& Reiter(2024)Kazan and Reiter]{kazan2022assessing}
Zeki Kazan and Jerome~P Reiter.
\newblock Assessing statistical disclosure risk for differentially private, hierarchical count data, with application to the 2020 us decennial census.
\newblock \emph{Statistica Sinica}, 34:\penalty0 to appear, 2024.

\bibitem[Kazan et~al.(2023)Kazan, Shi, Groce, and Bray]{kazan2023test}
Zeki Kazan, Kaiyan Shi, Adam Groce, and Andrew~P Bray.
\newblock The test of tests: A framework for differentially private hypothesis testing.
\newblock In \emph{International Conference on Machine Learning}, pp.\  16131--16151. PMLR, 2023.

\bibitem[Kifer \& Machanavajjhala(2014)Kifer and Machanavajjhala]{kifer2014pufferfish}
Daniel Kifer and Ashwin Machanavajjhala.
\newblock Pufferfish: A framework for mathematical privacy definitions.
\newblock \emph{ACM Transactions on Database Systems (TODS)}, 39\penalty0 (1):\penalty0 1--36, 2014.

\bibitem[Kifer et~al.(2022)Kifer, Abowd, Ashmead, Cumings-Menon, Leclerc, Machanavajjhala, Sexton, and Zhuravlev]{kifer2022bayesian}
Daniel Kifer, John~M Abowd, Robert Ashmead, Ryan Cumings-Menon, Philip Leclerc, Ashwin Machanavajjhala, William Sexton, and Pavel Zhuravlev.
\newblock Bayesian and frequentist semantics for common variations of differential privacy: Applications to the 2020 census.
\newblock \emph{arXiv preprint arXiv:2209.03310}, 2022.

\bibitem[Kruschke(2015)]{kruschke2014doing}
John~K Kruschke.
\newblock Bayesian approaches to testing a point (“null”) hypothesis.
\newblock In John~K Kruschke (ed.), \emph{Doing Bayesian Data Analysis (Second Edition)}, pp.\  335--358. Academic Press, Boston, second edition edition, 2015.
\newblock ISBN 978-0-12-405888-0.
\newblock \doi{https://doi.org/10.1016/B978-0-12-405888-0.00012-X}.
\newblock URL \url{https://www.sciencedirect.com/science/article/pii/B978012405888000012X}.

\bibitem[Li \& Reiter(2022)Li and Reiter]{li2022bayesian}
Linlin Li and Jerome~P Reiter.
\newblock Bayesian inference for estimating subset proportions using differentially private counts.
\newblock \emph{Journal of Survey Statistics and Methodology}, 10\penalty0 (3):\penalty0 785--803, 2022.

\bibitem[McClure \& Reiter(2012)McClure and Reiter]{mcclure2012differential}
David McClure and Jerome~P Reiter.
\newblock Differential privacy and statistical disclosure risk measures: An investigation with binary synthetic data.
\newblock \emph{Trans. Data Priv.}, 5\penalty0 (3):\penalty0 535--552, 2012.

\bibitem[Meredith \& Kruschke(2020)Meredith and Kruschke]{hdi}
Mike Meredith and John Kruschke.
\newblock \emph{HDInterval: Highest (Posterior) Density Intervals}, 2020.
\newblock URL \url{https://CRAN.R-project.org/package=HDInterval}.
\newblock R package version 0.2.2.

\bibitem[Minami et~al.(2016)Minami, Arai, Sato, and Nakagawa]{minami2016differential}
Kentaro Minami, Hitomi Arai, Issei Sato, and Hiroshi Nakagawa.
\newblock Differential privacy without sensitivity.
\newblock \emph{Advances in Neural Information Processing Systems}, 29, 2016.

\bibitem[Mortada et~al.(2001)Mortada, Sobh, El-Defrawy, and Farahat]{mortada2001study}
Wael Mortada, Mohamed Sobh, M~M El-Defrawy, and S~E Farahat.
\newblock Study of lead exposure from automobile exhaust as a risk for nephrotoxicity among traffic policemen.
\newblock \emph{American Journal of Nephrology}, 21\penalty0 (4):\penalty0 274--279, 2001.

\bibitem[Novomestky \& Nadarajah(2016)Novomestky and Nadarajah]{truncdist}
Frederick Novomestky and Saralees Nadarajah.
\newblock \emph{truncdist: Truncated Random Variables}, 2016.
\newblock URL \url{https://CRAN.R-project.org/package=truncdist}.
\newblock R package version 1.0-2.

\bibitem[Park \& Casella(2008)Park and Casella]{park2008bayesian}
Trevor Park and George Casella.
\newblock The {B}ayesian lasso.
\newblock \emph{Journal of the American Statistical Association}, 103\penalty0 (482):\penalty0 681--686, 2008.

\bibitem[Pe{\~n}a \& Barrientos(2024)Pe{\~n}a and Barrientos]{pena2022differentially}
V{\'\i}ctor Pe{\~n}a and Andr{\'e}s~F Barrientos.
\newblock Differentially private hypothesis testing with the subsampled and aggregated randomized response mechanism.
\newblock \emph{Statistica Sinica}, 34:\penalty0 to appear, 2024.

\bibitem[Schein et~al.(2019)Schein, Wu, Schofield, Zhou, and Wallach]{schein2019locally}
Aaron Schein, Zhiwei~Steven Wu, Alexandra Schofield, Mingyuan Zhou, and Hanna Wallach.
\newblock Locally private {B}ayesian inference for count models.
\newblock In \emph{International Conference on Machine Learning}, pp.\  5638--5648. PMLR, 2019.

\bibitem[Sheffet(2017)]{sheffet2017differentially}
Or~Sheffet.
\newblock Differentially private ordinary least squares.
\newblock In \emph{International Conference on Machine Learning}, pp.\  3105--3114. PMLR, 2017.

\bibitem[Sun \& Berger(2007)Sun and Berger]{sun2007objective}
Dongchu Sun and James~O Berger.
\newblock Objective {B}ayesian analysis for the multivariate normal model.
\newblock \emph{Bayesian Statistics}, 8:\penalty0 525--562, 2007.

\bibitem[Wang(2018)]{wang2018revisiting}
Yu-Xiang Wang.
\newblock Revisiting differentially private linear regression: optimal and adaptive prediction \& estimation in unbounded domain.
\newblock \emph{arXiv preprint arXiv:1803.02596}, 2018.

\bibitem[Wang et~al.(2015)Wang, Fienberg, and Smola]{wang2015privacy}
Yu-Xiang Wang, Stephen Fienberg, and Alex Smola.
\newblock Privacy for free: Posterior sampling and stochastic gradient {M}onte {C}arlo.
\newblock In \emph{International Conference on Machine Learning}, pp.\  2493--2502. PMLR, 2015.

\bibitem[Zhang et~al.(2012)Zhang, Zhang, Xiao, Yang, and Winslett]{zhang2012functional}
Jun Zhang, Zhenjie Zhang, Xiaokui Xiao, Yin Yang, and Marianne Winslett.
\newblock Functional mechanism: regression analysis under differential privacy.
\newblock \emph{Proceedings of the VLDB Endowment}, 5\penalty0 (11):\penalty0 1364--1375, 2012.

\bibitem[Zhang \& Zhang(2023)Zhang and Zhang]{zhang2023dp}
Wanrong Zhang and Ruqi Zhang.
\newblock Dp-fast mh: Private, fast, and accurate {M}etropolis-{H}astings for large-scale {B}ayesian inference.
\newblock In \emph{International Conference on Machine Learning}, pp.\  41847--41860. PMLR, 2023.

\bibitem[Zhang et~al.(2016)Zhang, Rubinstein, and Dimitrakakis]{zhang2016differential}
Zuhe Zhang, Benjamin Rubinstein, and Christos Dimitrakakis.
\newblock On the differential privacy of {B}ayesian inference.
\newblock In \emph{Proceedings of the AAAI Conference on Artificial Intelligence}, volume~30, 2016.

\end{thebibliography}
